\documentclass[preprint,12pt]{elsarticle}

\usepackage{amsmath,amssymb,amsfonts}
\usepackage{amsthm}

\usepackage{graphicx}
\usepackage{booktabs}
\usepackage{float}
\usepackage[section]{placeins} 

\usepackage{xcolor}
\usepackage[protrusion=true,expansion=false]{microtype}
\usepackage{tikz}
\usetikzlibrary{arrows,shapes,patterns,positioning,calc,arrows.meta}

\usepackage{hyperref}
\hypersetup{colorlinks=true,citecolor=blue,urlcolor=blue,linkcolor=blue}

\theoremstyle{plain}
\newtheorem{theorem}{Theorem}[section]
\newtheorem{corollary}[theorem]{Corollary}
\newtheorem{proposition}[theorem]{Proposition}
\theoremstyle{definition}
\newtheorem{example}{Example}[section]

\theoremstyle{remark}

\journal{Mathematics and Computers in Simulation}

\begin{document}

\begin{frontmatter}

\title{Numerical Reliability of Logistic Gene Regulatory Network Models:
       Preventing Expression Shutdown and Robust Integration of
       Boolean-Derived ODE Systems}

\author[inst1]{Ismail Belgacem}
\ead{ismail.belgacem.81@gmail.com}

\affiliation[inst1]{organization={Independent Researcher},
  addressline={Mezaourou},
  city={Ghazaouet},
  postcode={13421},
  state={Tlemcen},
  country={Algeria}}

\begin{abstract}
Gene regulatory networks are routinely translated from Boolean update rules
into large continuous ODE systems that must be integrated numerically for
attractor identification, parameter sensitivity analysis, and control design.
The reliability of that integration depends critically on the sigmoidal kernel
chosen to represent regulation. This paper presents a simulation study showing
that the Hill function---the near-universal choice---is a generically
unreliable kernel for this purpose, and that the logistic function is a robust
replacement. Two failure modes are demonstrated computationally. First, because
the Hill function vanishes identically at zero input, bistable circuits acquire
an absorbing off-state: using experimentally grounded \textit{Escherichia coli}
parameters for the galactose-operon positive-autoregulation motif, a Hill-based
model remains permanently trapped below the unstable separatrix, whereas the
logistic model---whose strictly positive basal rate is built into the
sigmoid---escapes the off-state in approximately $44$~minutes through basal
production alone, consistent with a conservative analytical estimate of
${\approx}58$~minutes and with \textit{gal}-operon induction kinetics. A
saddle-node analysis characterises the bistable parameter window through an
explicit transcendental equation and identifies the threshold $\lambda\theta=2$
separating monostable from bistable regimes. Second, when the Hill exponent is
non-integer---as it routinely is when fitting dose-response data---the power-law
$x^n$ requires evaluating $e^{n\ln x}$, which becomes complex-valued whenever an
adaptive solver overshoots into negative concentrations. On an $80$-gene
Boolean-derived benchmark with $n\approx3.509$, the Hill solver enters silent
complex-valued contamination from $t\approx52.64$ and produces smooth but
spurious trajectories with no visible artefact, whereas the logistic
formulation completes $t\in[0,200]$ without a single warning. Because the logistic right-hand side is
globally Lipschitz with a constant explicit in the network parameters, we
further establish an a priori global-error bound of classical order for
logistic-based integration---a convergence guarantee structurally
unavailable to the Hill formulation. The framework
is immediately deployable through standard numerical libraries and is natively
compatible with automatic-differentiation tools.
\end{abstract}


\end{frontmatter}

\section{Introduction}
\label{sec:introduction}

Mathematical modelling of gene regulatory networks (GRNs) has become an
indispensable tool for understanding the dynamical behaviours that emerge from
transcriptional interactions and for guiding the rational design of synthetic
gene circuits. A standard workflow translates a Boolean description of the
network---a set of logical update rules over genes and their
negations~\citep{kauffman1969metabolic,albert2003topology}---into a continuous
system of ordinary differential equations, which is then integrated numerically
to identify attractors, perform parameter sensitivity analysis, and design
feedback control~\citep{belgacem2020control,chambon2020qualitative}. As the networks of interest grow to tens or hundreds of
genes, the numerical integration of the resulting high-dimensional ODE
systems~\citep{belgacem2025glass,farcot2019chaos} becomes the computational bottleneck, and its reliability becomes a
first-order concern.

That reliability depends critically on the sigmoidal kernel used to represent
each regulatory interaction. The near-universal choice is the Hill function,
with activation form $h^+(x,\theta,n) = x^n/(x^n+\theta^n)$ and repression
form $h^-(x,\theta,n) = \theta^n/(x^n+\theta^n)$. The Hill function is
intuitive and mechanistically grounded, but as a kernel for large-scale
numerical simulation it carries two liabilities that this paper makes precise
and demonstrates computationally.

The first liability is an \emph{absorbing off-state}. Because
$h^+(0,\theta,n)=0$ for every $n>0$, any gene whose sole activator is
momentarily absent is assigned a production rate of exactly zero. In a bistable
circuit this is not a transient inconvenience: the off-state becomes an
absorbing set from which no intrinsic dynamics can escape. This is also
biologically unfaithful. Experimental studies across bacterial operons,
eukaryotic promoters, and synthetic circuits consistently show that genes are
never fully silent: even under strong repression, minimal transcription
persists at $0.1$--$10\%$ of maximal expression, and this basal activity
prevents irreversible shutdown and primes the transcriptional machinery for
rapid induction~\citep{becskei2000engineering,lipshtat2006genetic,%
weickert1993galactose}. In the \textit{Escherichia coli} galactose operon,
leaky transcription primes the positive-feedback loop for rapid activation;
without it, cells remain locked in the off-state for
hours~\citep{weickert1993galactose,hua1974multiple}. A Hill-based model cannot
represent this mechanism, and at the network scale every randomly initialised
target whose activator starts below threshold is frozen from the first
integration step.

The second liability is \emph{silent numerical corruption}. When the
cooperativity exponent $n$ is non-integer---as it routinely is when fitting
experimental dose-response data, yielding values such as $n\approx1.39$,
$2.73$, or $3.52$~\citep{gottschalk2005five,reeve2013pharmacodynamic,%
santillan2008use}---the power-law $x^n$ must be evaluated as $e^{n\ln x}$,
which is not a real number once a state component overshoots into a negative
concentration---as adaptive solvers routinely do when a trajectory approaches
zero. In a computing environment that evaluates such fractional powers on the
complex principal branch, the solver then continues to integrate a
complex-valued surrogate system, producing smooth and visually plausible
trajectories that are not solutions of the true biological model, with no
visual artefact to alert the modeller. For $n\in(k,k+1)$ the function is only $C^k$-smooth, so the
high-order convergence guarantees of Runge--Kutta and multistep methods do not
apply near the origin.

The logistic function removes both liabilities. Its activation form
$f^+(x,\theta,\lambda) = 1/(1+e^{-\lambda(x-\theta)})$ and repression form
$f^-(x,\theta,\lambda) = 1/(1+e^{\lambda(x-\theta)})$ are globally $C^\infty$,
real-valued for every argument including negative ones, and strictly positive
at zero concentration for all finite $\lambda$ and $\theta$. Basal expression
is therefore built into the shape of the function rather than appended to it,
and the right-hand side of a logistic-based ODE system is globally Lipschitz on
all of $\mathbb{R}^N$, so standard stability theory applies everywhere,
including near and below zero.

This paper is a simulation study of the consequences of this choice. Building
on the product-of-logistics modelling framework established in the companion
paper~\citep{belgacem2026framework}---whose essential definitions are
recalled in Section~\ref{sec:prelim} so that the present paper is
self-contained---we proceed as follows.
Section~\ref{sec:biological} examines the prevention of expression shutdown in
noisy low-expression regimes: for the galactose-operon positive-autoregulation
motif with experimentally grounded \textit{E.~coli} parameters, numerical
simulation shows the logistic model escaping the off-state in approximately
$44$~minutes through basal production alone, while the Hill model remains
permanently confined; Theorem~\ref{thm:bistability} characterises the
saddle-node set through an explicit transcendental equation, and
Corollary~\ref{cor:autoreg_stability} classifies the stability of every
equilibrium; a small Boolean-derived network then closes the section,
exhibiting the same loss of basal expression at the network level.
Section~\ref{sec:numerical_comparison} reports the central computational
experiment: a four-stage simulation protocol that constructs an $80$-gene
Boolean network, translates it automatically into a continuous ODE system,
integrates the system, and extracts state snapshots. Under identical conditions
the Hill solver enters silent complex-valued contamination at $t\approx52.64$
and terminates near $t\approx63$--$65$, leaving roughly two-thirds of the
intended horizon covered only by unconstrained extrapolation, whereas the
logistic formulation completes $t\in[0,200]$ without a single warning.
Section~\ref{sec:numerical_comparison} concludes with a convergence
proposition (Proposition~\ref{prop:apriori}): the logistic right-hand side
meets the hypotheses of the classical error theory of one-step methods with
constants explicit in the network parameters, whereas the Hill right-hand
side structurally cannot; a fixed-step convergence study confirms that the
logistic system attains this guaranteed order in practice.
The framework requires only standard numerical integration libraries and is
natively compatible with automatic-differentiation tools.

\section{Preliminaries: The Logistic Modelling Framework}
\label{sec:prelim}

This section recalls, without proof, the elements of the product-of-logistics
modelling framework required for the numerical study that follows. Full
derivations, the stability analysis of the two-gene oscillator, the structural
properties of the De~Morgan map, and the well-posedness theory are given in the
companion paper~\citep{belgacem2026framework}; a broader exploration of
logistic alternatives to Hill functions in genetic-network modelling is
given in~\cite{belgacem2025exploring}.

\paragraph{Regulatory kernels.}
Throughout, regulation is represented by sigmoidal functions of protein
concentration. The \emph{logistic} activation and repression kernels are
\begin{equation}
  f^+(x,\theta,\lambda) = \frac{1}{1+e^{-\lambda(x-\theta)}},
  \qquad
  f^-(x,\theta,\lambda) = \frac{1}{1+e^{\lambda(x-\theta)}}
  = 1-f^+(x,\theta,\lambda),
  \label{eq:p2_logistic_kernels}
\end{equation}
with threshold $\theta>0$ and steepness $\lambda>0$; both are globally
$C^\infty$, real-valued on all of $\mathbb{R}$, and strictly positive at
$x=0$. For comparison, the \emph{Hill} kernels are
$h^+(x,\theta,n) = x^n/(x^n+\theta^n)$ and
$h^-(x,\theta,n) = \theta^n/(x^n+\theta^n)$, defined for $x\ge0$, with
$h^+(0,\theta,n)=0$ for every $n>0$.

\paragraph{Single-gene dynamics.}
For a network of $N$ genes, $x_i(t)$ denotes the concentration of the protein
product of gene $i$. Each gene obeys a balance between synthesis and
degradation,
\begin{equation}
  \dot{x}_i \;=\; \kappa_i\,\Phi_i(\mathbf{x}) \;-\; \gamma_i\,x_i,
  \label{eq:ode_0}
\end{equation}
where $\kappa_i>0$ is the maximal production rate, $\gamma_i>0$ the degradation
rate, and $\Phi_i\colon\mathbb{R}^{N}\to[0,1]$ the regulatory function that
synthesises all activating and repressing influences on gene $i$.

\paragraph{Multi-gene product-of-logistics system.}
When gene $i$ is regulated by a set $\mathcal{A}_i$ of activators and a set
$\mathcal{R}_i$ of repressors acting through independent binding sites, the
regulatory function is the product
\begin{equation}
  \dot{x}_i
  \;=\;
  \kappa_i
  \!\left(
    \prod_{j \in \mathcal{A}_i}
      \frac{1}{1 + e^{-\lambda(x_j - \theta_{ij})}}
    \;\cdot\;
    \prod_{k \in \mathcal{R}_i}
      \frac{1}{1 + e^{-\lambda(\theta_{ik} - x_k)}}
  \right)
  - \gamma_i x_i ,
  \label{eq:multi_gene_system}
\end{equation}
in which each activator contributes an increasing logistic factor and each
repressor a decreasing one. The right-hand side of~\eqref{eq:multi_gene_system}
is globally $C^\infty$ and globally Lipschitz on $\mathbb{R}^{N}$ with the
explicit constant
$L=\max_i\bigl(\gamma_i+\kappa_i\lambda(|\mathcal{A}_i|+|\mathcal{R}_i|)/4\bigr)$,
and the box $\prod_i[0,\kappa_i/\gamma_i]$ is forward-invariant; the system is
therefore globally well-posed~\citep{belgacem2026framework}.

\paragraph{Translating Boolean rules.}
A Boolean update rule is mapped to a continuous regulatory function $\Phi$ by
sending each positive literal to an increasing logistic, each negative literal
to a decreasing logistic, and each conjunction to a product. A disjunction
$C_1(\mathbf{x})\vee\cdots\vee C_m(\mathbf{x})$ is mapped through the recursive
De~Morgan product formula
\begin{equation}
  \Phi\!\left(\bigvee_{k=1}^{m} C_k(\mathbf{x})\right)
  \;=\;
  1 - \prod_{k=1}^{m}\bigl(1 - \Phi(C_k(\mathbf{x}))\bigr),
  \label{eq:demorgan_0}
\end{equation}
which keeps $\Phi\in[0,1]$ regardless of the number of independent regulatory
pathways. Formula~\eqref{eq:demorgan_0} is the continuous counterpart of the
classical De~Morgan law and coincides with the probability that at least one of
$m$ independent events occurs.

\paragraph{The two-gene logistic oscillator.}
The canonical negative-feedback motif, in which gene~1 is repressed by gene~2
and gene~2 is activated by gene~1, is governed by
\begin{equation}
\begin{aligned}
\dot{x}_1 &= \kappa_1\, f^-(x_2,\theta_2,\lambda) - \gamma_1 x_1
        \;=\; \kappa_1\, \frac{1}{1 + e^{\lambda (x_2 - \theta_2)}} - \gamma_1 x_1, \\[4pt]
\dot{x}_2 &= \kappa_2\, f^+(x_1,\theta_1,\lambda) - \gamma_2 x_2
        \;=\; \kappa_2\, \frac{1}{1 + e^{-\lambda (x_1 - \theta_1)}} - \gamma_2 x_2 .
\end{aligned}
\label{eq:oscillator}
\end{equation}
This planar system is globally asymptotically stable at its unique equilibrium
and admits no Hopf bifurcation, so sustained oscillations require explicit time
delays~\citep{belgacem2026framework,belgacem2026sustained,belgacem2026beyond}; it
serves below only as a low-expression test case.

\section{Biological Realism Through Low-Expression Modelling}
\label{sec:biological}

Real gene regulatory systems never fully shut down. Even under strong repression, promoters
exhibit low-level transcription, often referred to as ``promoter leakiness'' or basal activity,
driven by factors such as nucleosome positioning, stochastic pre-initiation complex formation,
and incomplete repressor binding. This basal expression is ubiquitous in GRNs, serving critical
functions: it reduces phenotypic noise by shifting gene expression distributions from multimodal
(favouring adaptive heterogeneity) to unimodal (promoting uniform responses), and it prevents
systems from trapping in irreversible off-states during stochastic fluctuations. The logistic
function naturally captures this imperfect inhibition. Unlike Hill functions, which drop to zero
at low input and can lock systems into unresponsive states in bistable or feedback motifs,
logistic models maintain a small but nonzero production rate, allowing noise-driven escapes and
rapid recovery. In resource-limited cellular environments, where maximal expression may be
unattainable even in the absence of repressors, the logistic's saturation reflects capacity
constraints rather than binding cooperativity, thereby enhancing biological realism in
large-scale GRNs.

\subsection{Preventing Shutdown in Noisy Biological Environments}

The always-positive production rate of logistic models ensures that systems remain responsive
to small perturbations, a feature crucial in noisy cellular environments. We illustrate this
advantage through two canonical examples: genetic oscillators and autoregulatory networks.

\begin{example}[Genetic Oscillator]
Consider a two-gene negative feedback loop in which the first gene activates the second, which
then represses the first. The Hill-based model is:
\begin{equation}
\begin{aligned}
\dot{x}_1 &= \kappa_1 \frac{\theta_2^n}{x_2^n + \theta_2^n} - \gamma_1 x_1, \\
\dot{x}_2 &= \kappa_2 \frac{x_1^n}{x_1^n + \theta_1^n} - \gamma_2 x_2.
\end{aligned}
\label{eq:oscillator_Hill}
\end{equation}
The logistic counterpart is given by Eq.~\eqref{eq:oscillator}.

To understand the difference, consider the system's behaviour starting from near-zero initial
conditions (\( x_1(0) \approx 0.02 \), \( x_2(0) \approx 0.02 \)) under strong degradation
(\(\gamma_1 = 8.0\), \(\gamma_2 = 5.0\)), with \( \lambda = n = 3 \), \(\kappa_1 = \kappa_2 =
0.5\), and \(\theta_1 = \theta_2 = 1.0\). Simulations are shown in
Figure~\ref{fig:Oscillateur_original_l_H}.

\begin{figure}[t]
    \centering
    \includegraphics[width=0.45\linewidth]{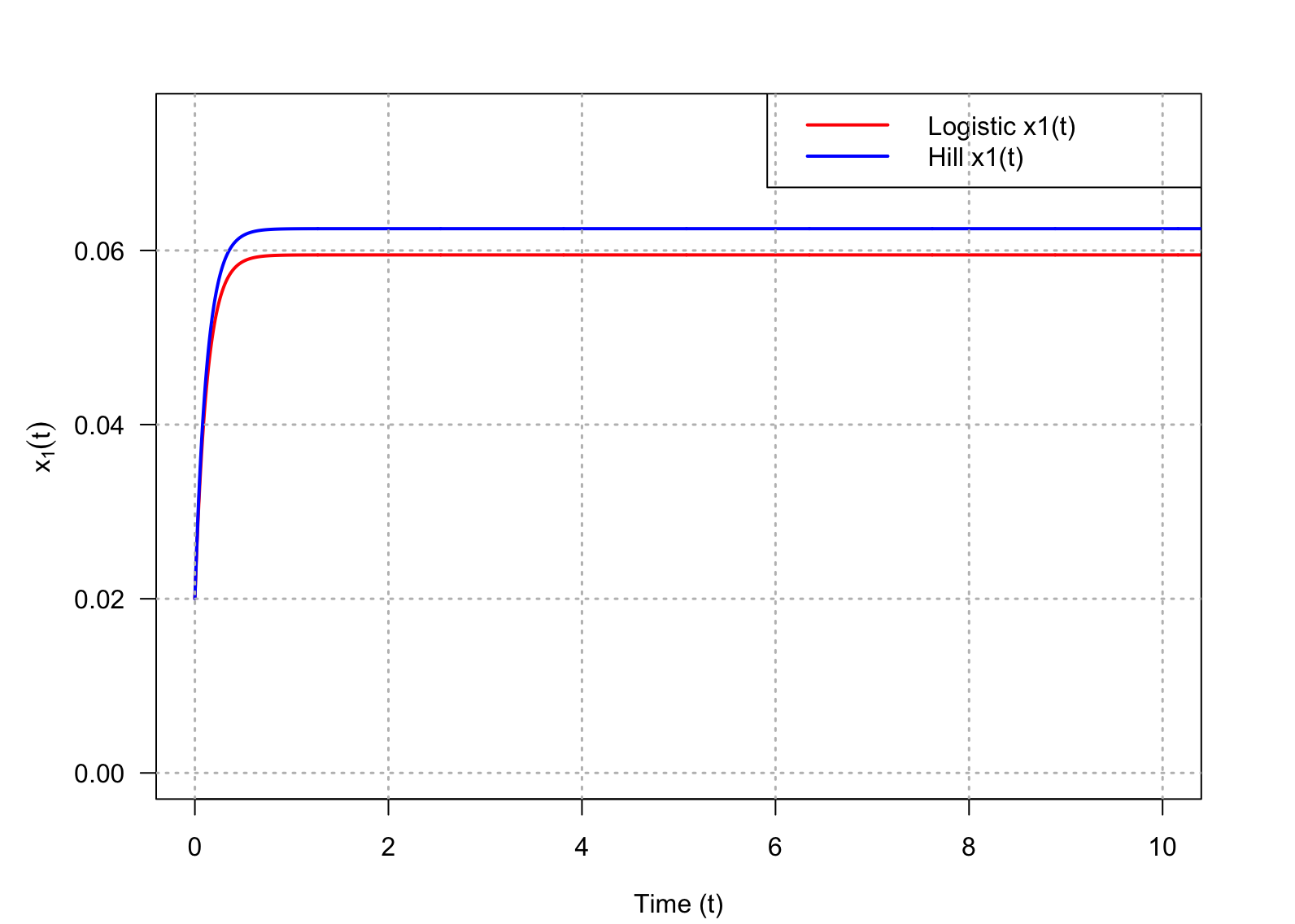}
    \includegraphics[width=0.45\linewidth]{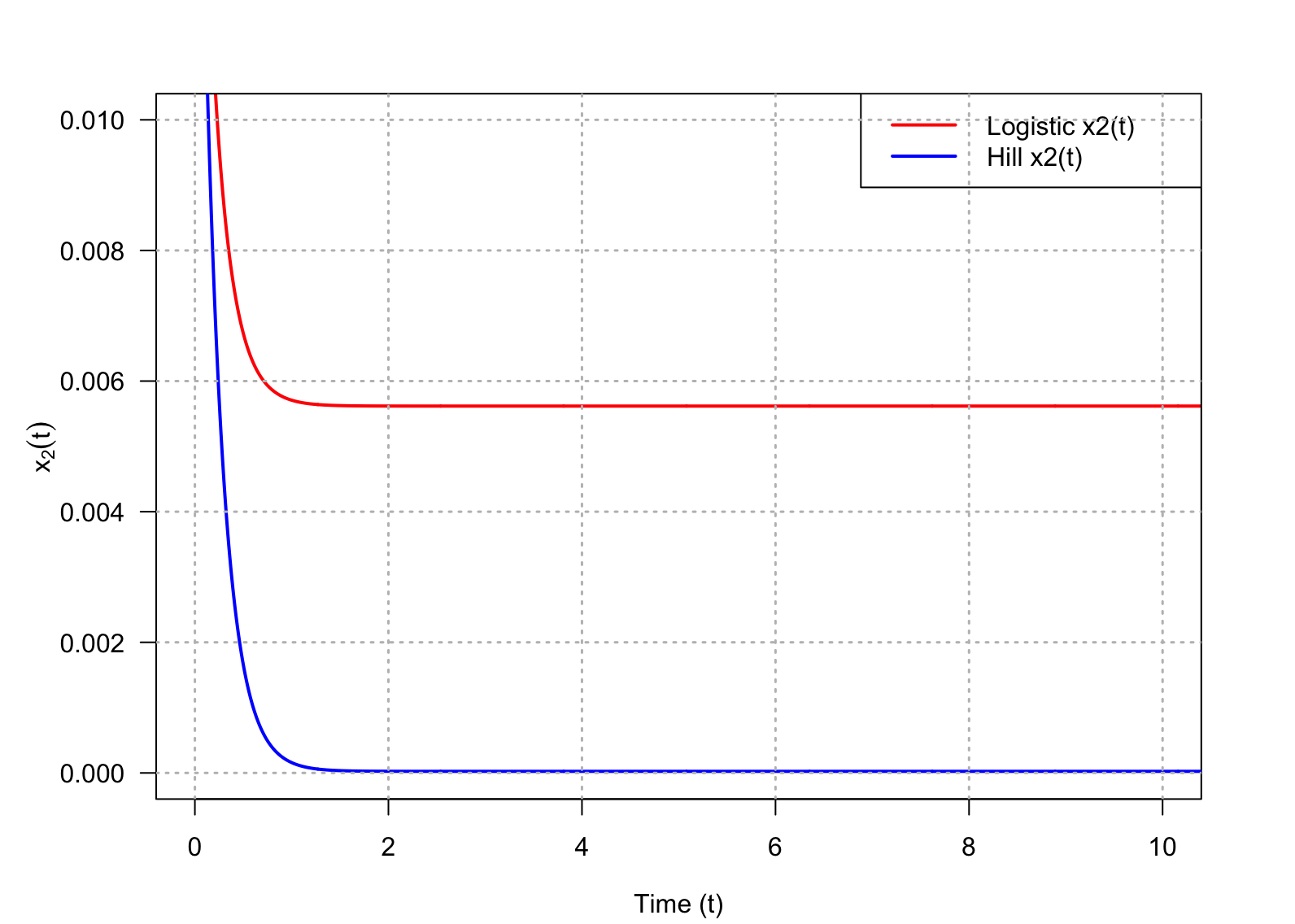}
    \caption{Trajectories of the genetic oscillator under small
    perturbations and strong degradation, comparing the logistic model
    (red) and the Hill model (blue). Left: \(x_1(t)\); Right: \(x_2(t)\).
    The contrast is sharpest in \(x_2(t)\): the logistic model retains a
    small but strictly nonzero steady expression level, whereas the Hill
    model collapses to a near-zero level, trapped by its vanishing basal
    production. Parameters: \(\lambda = n = 3\), \(\kappa_1=\kappa_2=0.5\),
    \(\gamma_1=8.0\), \(\gamma_2=5.0\), \(\theta_1=\theta_2=1.0\); initial
    conditions \(x_1(0)=x_2(0)=0.02\).}
    \label{fig:Oscillateur_original_l_H}
\end{figure}

The Hill model's zero response at zero input creates a critical vulnerability. When \( x_1
\approx 0 \), the production rate for \( x_2 \) is effectively zero: \(\kappa_2 \cdot 0 = 0\).
This complete shutdown traps the system in a low-expression state. For moderate cooperativity (\( n
= 3 \)), the Hill function is still steep enough to require a substantial increase in \( x_1 \) to generate
appreciable \( x_2 \). In noisy environments, small stochastic fluctuations in \( x_1 \) may be
insufficient to escape this off-state, delaying activation or disrupting oscillations entirely.

The logistic model, by contrast, maintains a positive production rate even
when the repressor $x_2$ is strongly expressed. The decreasing logistic
$f^-(x_2,\theta_2,\lambda) = 1/(1+e^{\lambda(x_2-\theta_2)})$ never reaches
zero: even at large repressor concentrations $x_2 \gg \theta_2$,
\[
\kappa_1 \cdot \frac{1}{1+e^{\lambda(x_2-\theta_2)}} > 0,
\]
ensuring that $x_1$ retains a small but persistent production rate that
cannot be entirely abolished by repression. Symmetrically, in the absence
of activator $x_1$, the increasing logistic $f^+(x_1,\theta_1,\lambda) =
1/(1+e^{-\lambda(x_1-\theta_1)})$ provides a basal activation level
$1/(1+e^{\lambda\theta_1}) > 0$, so production of $x_2$ continues at low
rate even when $x_1 = 0$. This slight but persistent production allows the
system to respond to small perturbations: if $x_1$ rises even modestly,
$x_2$ can respond, facilitating recovery or sustained oscillatory
dynamics. As Figure~\ref{fig:Oscillateur_original_l_H} demonstrates, the
logistic model escapes the low-expression trap, while the Hill model
remains stuck.

This behaviour aligns with experimental findings. Joanito et al.~\cite{joanito2020basal} examined
synthetic oscillators and demonstrated that basal transcriptional leakage has complex, model-dependent
effects on oscillatory dynamics. Their stochastic simulations showed that the impact of leakage
depends critically on whether control operates at the transcriptional or post-translational level:
in models with combined transcriptional and post-translational controls, a leakage level of
20\% of maximum expression enabled 38.1\% of parameter sets to sustain oscillation.
Independently, their analysis establishes that maintaining a minimum non-zero production rate
is essential for enabling gene circuits to escape low-expression traps under noise---precisely
the mechanism demonstrated here for the logistic model.
\end{example}

\subsection{Positive Autoregulation: Bistability and Basal Expression}
\label{ex:positive_autoregulation}

Positive autoregulation, in which a transcription factor activates its own gene expression,
exemplifies the profound interplay among low expression levels, molecular noise, and bistability
in gene regulatory networks (GRNs). This motif is not merely a theoretical construct but appears
pervasively in biological systems, most notably in the \textit{Escherichia coli} galactose
(\textit{gal}) operon, where regulators such as GalS self-activate to fine-tune metabolic
responses to environmental sugar availability~\citep{weickert1993galactose,hua1974multiple}. The
architectural elegance of this feedback loop lies in its dual nature: it enhances both response
speed and robustness to perturbations, yet simultaneously introduces bistability, the coexistence
of two stable steady states (high and low expression levels) separated by an unstable threshold.

The biological significance of bistability in positive autoregulation cannot be overstated. In
cellular environments characterised by stochasticity arising from low molecule counts, often
on the order of tens to a few hundred transcription factor molecules per cell, the system faces a critical
challenge: under low initial conditions or in the presence of stochastic fluctuations, the
regulatory circuit can become irreversibly trapped in an ``off'' state~\citep{becskei2000engineering}.
This phenomenon has profound implications for cellular decision-making, metabolic switching, and
developmental processes. Without intrinsic recovery mechanisms, such as basal (leaky)
transcription, cells may fail to respond appropriately to environmental cues, potentially
leading to metabolic failure or developmental arrest.

The choice of mathematical model used to describe transcriptional regulation profoundly impacts
our ability to capture these dynamics. Hill functions, widely employed for their sigmoidal shape
that mimics cooperative binding of transcription factors to DNA, make a critical simplifying
assumption: they posit zero basal expression in the absence of an activator. This assumption,
while mathematically convenient and appropriate for describing sharp, switch-like responses,
introduces both numerical instability and a fundamental failure to capture noise-driven escape
from low-expression states. In stark contrast, logistic functions possess an inherent non-zero
basal production rate, which more accurately reflects in vivo conditions where leaky
transcription occurs due to incomplete repression, constitutive (albeit weak) promoter activity,
or stochastic binding events~\citep{lipshtat2006genetic}. This distinction becomes particularly
critical in high-degradation or noisy cellular environments, where Hill-based models require the
ad hoc addition of leakage terms to match experimental observations, a limitation repeatedly
documented in both synthetic gene circuits and natural operons such as the \textit{gal}
system~\citep{becskei2000engineering}.

To capture the biological reality of gene expression, we employ a two-level mathematical
framework that explicitly separates mRNA (\(m\)) and protein (\(x\)) dynamics. This approach
offers substantially greater biological accuracy than single-variable approximations, which
implicitly assume instantaneous translation or quasi-steady-state mRNA levels. The separation is
justified by the distinct time scales: mRNA molecules typically have half-lives on the order of
minutes (rapid turnover), whereas proteins, particularly transcription factors in bacteria, can
persist for hours or even exhibit negligible degradation relative to cellular division
(dilution-dominated decay).

The logistic model for positive autoregulation is given by:
\begin{equation}
\begin{aligned}
\dot{m} &= k_m \frac{1}{1 + e^{-\lambda (x - \theta)}} - k_{dm} m, \\
\dot{x} &= k_p m - k_{dp} x,
\end{aligned}
\label{eq:autoreg_logistic}
\end{equation}
where the parameters carry the following biological interpretations. The maximum mRNA synthesis rate $k_m$ (in molecules per second) encapsulates promoter strength, RNA polymerase binding efficiency, and elongation rate, while $k_{dm}$ is the first-order mRNA decay rate, reflecting both enzymatic degradation by RNases and dilution due to cell growth. The protein synthesis rate per mRNA $k_p$ represents the translation efficiency (rate at which ribosomes translate each mRNA molecule), and $k_{dp}$ is the protein degradation or dilution rate, dominated by cell-division dilution for stable proteins. The steepness parameter $\lambda$ controls the sigmoidicity of the activation response and is analogous to the Hill coefficient $n$, reflecting the effective cooperativity of transcription-factor binding~\citep{becskei2000engineering}; the activation threshold $\theta$ is the protein concentration at which $f(\theta)=1/2$, corresponding to the effective dissociation constant or half-activation point~\citep{lipshtat2006genetic}.

The logistic activation function \(f(x) = \frac{1}{1 + e^{-\lambda (x - \theta)}}\) smoothly
interpolates between a basal production rate at low protein levels and maximal production at
saturation. Critically, at \(x = 0\) (complete absence of protein), the basal production rate
ensures continuous, albeit low-level, mRNA synthesis even in the absence of autoactivation. For
comparison, the corresponding Hill-based model is:
\begin{equation}
\begin{aligned}
\dot{m} &= k_m \frac{x^n}{x^n + c^n} - k_{dm} m, \\
\dot{x} &= k_p m - k_{dp} x,
\end{aligned}
\label{eq:autoreg_hill}
\end{equation}
where the Hill coefficient $n$ represents the effective number of transcription-factor binding sites or the degree of cooperative binding~\citep{santillan2008use}, and the half-activation constant $c$ is the protein concentration at which production reaches half-maximum ($h(c)=1/2$); the latter is often normalised to unity ($c=1$) for dimensionless analysis~\citep{cherry2000make}.

The Hill function \(h(x) = \frac{x^n}{x^n + c^n}\) exhibits the critical limitation that
\(h(0) = 0\), meaning that in the complete absence of protein, mRNA synthesis ceases entirely.
This creates a mathematical and biological problem: without any protein, the system cannot escape
from the zero state through its own dynamics, and recovery requires either large external
perturbations or the addition of ad hoc leakage terms.

The parameter values, summarised in Table~\ref{tab:parameters}, are derived from experimental measurements in \textit{E.~coli} and reflect the physiology of low-copy-number gene expression, where fast mRNA turnover contrasts sharply with persistent protein levels. The mRNA synthesis rate $k_m = 0.003$~molecules\,s$^{-1}$ is chosen from the lower end of the experimentally observed range $0.001$--$0.1$~molecules\,s$^{-1}$ for \textit{E.~coli} promoters in cell-free systems~\citep{marshall2019quantitative}, appropriate for average or regulated transcription in low-copy contexts ($1$--$10$ molecules). The mRNA degradation rate $k_{dm} = 0.001$~s$^{-1}$ corresponds to a half-life of about $11.5$~minutes, somewhat above the typical median of $5$--$7$~minutes reported under MOPS-glucose growth conditions~\citep{bernstein2002global} and emphasising the longer-lived transcripts characteristic of autoregulatory transcription factors. The protein synthesis rate $k_p = 0.002$~s$^{-1}$ lies well below the genome-wide mean translation-initiation rate of $0.1$--$0.3$~s$^{-1}$~\citep{li2014quantifying}, consistent with regulated low-expression genes. The protein degradation/dilution rate $k_{dp} = 0.00001$~s$^{-1}$ (half-life $\sim$$19$--$23$~hours) reflects the biological reality that many bacterial transcription factors are stable proteins whose levels are primarily controlled by dilution rather than active proteolysis~\citep{nath1970protein}; the choice is justified in detail below. Finally, the steepness/cooperativity is set to $\lambda = n = 3$ (moderate ultrasensitivity without extreme cooperativity) and the threshold is normalised to $\theta = c = 1$ for dimensionless analysis.

\begin{table}[!htbp]
\centering
\caption{Biophysical parameters for positive autoregulation model, grounded in \textit{E.~coli}
physiology and experimental measurements. Values are selected to represent low-copy, regulated
gene expression characteristic of transcription factor circuits.}
\small
\begin{tabular}{l c p{4.7cm} p{2.5cm}}
\toprule
Parameter & Value & Biological Context & Source \\
\midrule
\(k_m\) & 0.003~molec.\,s\(^{-1}\) & Maximum mRNA synthesis; selected from lower end of broad range
(0.001--0.1~molecules\,s\(^{-1}\)) for average/regulated promoters; reflects elongation \(\sim\)10--50~nt/s,
initiation \(\sim\)0.1--1/min & \citep{marshall2019quantitative} \\[0.3cm]
\(k_{dm}\) & 0.001~s\(^{-1}\) & mRNA degradation; half-life \(\sim\)11.5~min, from upper tail of
distribution (median 5--7~min; 80\% within 3--8~min; extremes \(<2\) to \(>30\)~min in
MOPS-glucose) & \citep{bernstein2002global} \\[0.3cm]
\(k_p\) & 0.002~s\(^{-1}\) & Protein synthesis per mRNA; below typical 0.1--0.3~s\(^{-1}\)
range for regulated genes; possibly adjusted for saturated translation at low mRNA copy numbers
& \citep{li2014quantifying} \\[0.3cm]
\(k_{dp}\) & 0.00001~s\(^{-1}\) & Protein degradation/dilution; half-life \(\sim\)19--23~hr;
approximated from slow component (0.5--1.5\%/hr) adjusted upward to \(\sim\)3.6\%/hr for
dilution in slow growth/stationary phase & \citep{nath1970protein} \\[0.3cm]
\(\lambda, n\) & 3 & Steepness/cooperativity; typical value for moderate ultrasensitivity in
bacterial transcription factors & \citep{becskei2000engineering} \\[0.3cm]
\(\theta, c\) & 1 & Threshold/half-activation; normalised affinity constant for dimensionless
analysis & \citep{lipshtat2006genetic} \\
\bottomrule
\end{tabular}
\label{tab:parameters}
\end{table}

The protein degradation/dilution rate $k_{dp} = 0.00001$~s$^{-1}$ warrants particular attention, as the corresponding half-life of approximately $19$--$23$~hours substantially exceeds typical bacterial protein turnover rates. Three biological considerations justify this choice. First, autoregulatory transcription factors, particularly those governing metabolic switching such as GalS in the \textit{gal} operon, display markedly longer half-lives than average cellular proteins: in their proteome-wide budding-yeast study, Belle et al.~\cite{belle2006quantification} found that proteins involved in regulatory roles tend to fall into a class of comparatively long-lived species, well above the genome-wide median, and an analogous pattern of regulator stability has been observed across organisms; a subset of stable regulators involved in bistable circuits exhibit even longer half-lives under slow-growth conditions where active proteolysis is downregulated. Second, in slowly growing or stationary-phase \textit{E.~coli}, protein levels are primarily controlled by dilution through cell division rather than by active proteolysis~\citep{nath1970protein}: under the nutrient-limitation or stationary-phase conditions relevant to metabolic switching, doubling times extend to $4$--$12$~hours or more, and cell division becomes the rate-limiting step for protein clearance, with the chosen $k_{dp}$ corresponding to a doubling-time-equivalent dilution rate of about $20$~hours. Third, the galactose-utilisation system---our biological archetype for positive autoregulation---operates precisely in these slow-growth or transition regimes: cells pre-grown on glucose and then shifted to galactose experience an extended lag phase ($1$--$3$~hours) during which the positive feedback loop must activate~\citep{weickert1993galactose,hua1974multiple}, and Weickert and Adhya~\cite{weickert1993galactose} demonstrated that GalS protein persists for multiple cell divisions after transcriptional shut-off, consistent with dilution-limited turnover rather than rapid proteolysis.

At steady state, the time derivatives in Equations~\eqref{eq:autoreg_logistic} and
\eqref{eq:autoreg_hill} vanish, yielding:
\begin{equation}
m^* = \frac{k_m}{k_{dm}} \, h(x^*), \qquad x^* = \frac{k_p}{k_{dp}} \, m^*,
\end{equation}
where \(h(x)\) denotes either the logistic function \(f(x)\) or the Hill function. Eliminating
the mRNA steady state \(m^*\) gives the fixed-point condition:
\begin{equation}
x^* = \underbrace{\frac{k_m k_p}{k_{dm} k_{dp}}}_{\displaystyle \alpha} \, h(x^*).
\label{eq:fixedpoint}
\end{equation}

\subsubsection{Biological Interpretation of $\alpha$}
The dimensionless parameter $\alpha$
represents the overall loop gain or feedback amplification:
\begin{equation}
\alpha = \frac{k_m k_p}{k_{dm} k_{dp}}
= \underbrace{\frac{k_m}{k_{dm}}}_{\substack{\text{transcriptional} \\ \text{gain}}}
\times
\underbrace{\frac{k_p}{k_{dp}}}_{\substack{\text{translational} \\ \text{gain}}}.
\end{equation}
The transcriptional gain \(k_m/k_{dm}\) represents the maximum steady-state mRNA level when the
promoter is fully activated. The translational gain \(k_p/k_{dp}\) represents the number of
protein molecules produced per mRNA at steady state. Their product \(\alpha\) thus quantifies
the maximum protein level achievable when activation is saturated (\(h(x) \to 1\)), giving
\(x_{\text{ss}} \approx \alpha\).

Taking these parameter values (above), the overall feedback amplification is:
\begin{equation}
\alpha = \frac{k_m k_p}{k_{dm} k_{dp}}
= \frac{(0.003)(0.002)}{(0.001)(0.00001)}
= \frac{6 \times 10^{-6}}{1 \times 10^{-8}} = 600.
\end{equation}
This value of \(\alpha \approx 600\) reflects the characteristic physiology of \textit{E.~coli}:
fast mRNA turnover coupled with exceptionally stable proteins, yielding a system with very high
loop gain.

\subsubsection{Bistability Analysis for the Hill Model}
For the Hill function $h(x) = \frac{x^n}{x^n + c^n}$,
bistability requires multiple intersections of the nullcline \(x = \alpha h(x)\) with the
identity line \(y = x\). For small \(\alpha\), there is only one intersection (monostable at low
expression). For large \(\alpha\), three intersections can occur: two stable fixed points (low
and high expression) separated by an unstable saddle point.

The critical condition for bistability arises at the tangency (saddle-node bifurcation) points,
where the nullcline becomes tangent to the identity line. At tangency:
\begin{equation}
x = \alpha h(x), \qquad 1 = \alpha h'(x).
\label{eq:tangency_hill}
\end{equation}
Eliminating \(\alpha\) gives:
\begin{equation}
\frac{x}{h(x)} = \frac{1}{h'(x)}.
\end{equation}
For the Hill function, $h'(x) = \frac{n c^n x^{n-1}}{(x^n + c^n)^2}$. Substituting $h$ and $h'$ into the tangency condition $x/h(x) = 1/h'(x)$ yields
\begin{equation}
x^n = (n-1)\,c^n,
\qquad
x_{\mathrm{crit}} = c\,(n-1)^{1/n}.
\end{equation}
Substituting $x_{\mathrm{crit}}$ back, the Hill value at the tangent is $h(x_{\mathrm{crit}}) = (n-1)/n$, and the critical amplification is
\begin{equation}
\alpha_{\mathrm{crit}}
= \frac{x_{\mathrm{crit}}}{h(x_{\mathrm{crit}})}
= \frac{n\,c}{(n-1)^{(n-1)/n}}.
\label{eq:alpha_crit_hill}
\end{equation}
For \(n = 3\) and \(c = 1\):
\begin{equation}
\alpha_{\mathrm{crit}}
= \frac{3}{2^{2/3}}
\approx \frac{3}{1.587}
\approx 1.89.
\end{equation}

Since our system has \(\alpha \approx 600 \gg 1.89\), the Hill model is well within the
\textbf{bistable} regime, supporting \textbf{two} stable steady states: the absorbing state
\(x = 0\) and a high-expression state near \(x \approx 600\), separated by an unstable
separatrix near \(x \approx 0.041\). Crucially, because \(h(0) = 0\), the origin
\((m, x) = (0, 0)\) is \emph{always} a stable fixed point of the Hill system: its Jacobian
eigenvalues \(-k_{dm}\) and \(-k_{dp}\) are both strictly negative for \emph{all} values of
\(\alpha\), regardless of how large \(\alpha\) becomes. The zero state is therefore never
eliminated by strong feedback: a trajectory starting at \(x(0) = 0.01\), which lies below the
unstable separatrix at \(x^* \approx 0.041\), decays irreversibly toward \(x = 0\). This is
precisely the trapping mechanism that the logistic model overcomes.

In summary, higher cooperativity $n$ lowers $\alpha_{\mathrm{crit}}$, meaning
less amplification is required to achieve bistability. Stronger ultrasensitivity enables
bistability with weaker feedback loops.

\subsubsection{Bistability Analysis for the Logistic Model}
For the logistic function
$f(x) = \frac{1}{1 + e^{-\lambda(x-\theta)}}$, the analysis is similar but yields a
\emph{finite} bistable range due to non-zero basal activity. The derivative
\begin{equation}
f'(x) = \lambda\,f(x)\,(1 - f(x))
\end{equation}
achieves its maximum \(\lambda/4\) at \(x = \theta\).

At saddle-node bifurcations:
\begin{equation}
x = \alpha f(x), \qquad 1 = \alpha f'(x).
\label{eq:tangency_logistic}
\end{equation}
Eliminating \(\alpha\):
\begin{equation}
x = \frac{1}{\lambda\,(1 - f(x))}.
\end{equation}
Setting \(y = f(x)\) and \(z = \lambda(x - \theta)\), so that
\(y = 1/(1 + e^{-z})\) and \(x = \theta + z/\lambda\), we obtain the transcendental equation:
\begin{equation}
\lambda\theta + z = 1 + e^{z}.
\label{eq:transcendental}
\end{equation}
For each root \(z\), the critical amplification is:
\begin{equation}
\alpha_{\mathrm{crit}}
= \frac{1}{\lambda\,y\,(1-y)}.
\end{equation}

For \(\lambda = 3\) and \(\theta = 1\), Equation~\eqref{eq:transcendental} becomes
\(e^z - z - 2 = 0\), with two roots. We order them by their \(\alpha_{\mathrm{crit}}\) values,
which is the biologically meaningful ordering:

The first root, $z \approx 1.1462$, gives $y \approx 0.7588$ and the lower critical amplification $\alpha_{\mathrm{crit}} \approx 1.821$, at which the high-expression state is born (upper-state saddle-node). The second root, $z \approx -1.8414$, gives $y \approx 0.1369$ and the upper critical amplification $\alpha_{\mathrm{crit}} \approx 2.821$, at which the low-expression state is annihilated (lower-state saddle-node).

Thus, the logistic model exhibits bistability in the range:
\begin{equation}
1.821 \;<\; \alpha \;<\; 2.821.
\end{equation}
The lower \(\alpha\) threshold marks the emergence of the high-expression state (saddle-node
bifurcation creating the high and unstable fixed points), while the upper \(\alpha\) threshold
marks the annihilation of the low-expression state once basal production overwhelms degradation.

With \(\alpha \approx 600 \gg 2.821\), our logistic system lies \textbf{above} the upper
bistability threshold and is therefore \textbf{monostable at the high state}: basal production
is so strong relative to degradation that, even starting from zero, the system is always driven
upward. We operate deliberately in this monostable-high regime in order to isolate the basal
production escape mechanism; bistability would require \(1.821 < \alpha < 2.821\).

Table~\ref{tab:alpha_crit_logistic} summarises critical amplification values for different
steepness parameters.

\begin{table}[!htbp]
\centering
\caption{Critical amplification \(\alpha_{\mathrm{crit}}\) for bistability in logistic-based
positive autoregulation with \(\theta = 1\). The bistable range \textbf{expands} as steepness
increases: the lower threshold decreases toward~1 (bistability requires less amplification)
while the upper threshold diverges to infinity (the low state persists over an ever-wider range
of \(\alpha\)). For \(\lambda = 2\), the product \(\lambda\theta = 2\) is the degenerate onset
boundary (single tangency at \(\alpha = 2\); no bistable range exists). For
\(\lambda \to \infty\), the logistic approaches a step function, recovering bistability for all
\(\alpha > 1\) with the upper threshold diverging.}
\begin{tabular}{ccc}
\toprule
Steepness \(\lambda\) & Lower \(\alpha_{\mathrm{crit}}\) & Upper \(\alpha_{\mathrm{crit}}\) \\
\midrule
2          & \multicolumn{2}{c}{Degenerate (\(\lambda\theta = 2\)); no bistable range} \\
3          & \(\approx 1.82\) & \(\approx 2.82\) \\
4          & \(\approx 1.68\) & \(\approx 5.28\) \\
5          & \(\approx 1.58\) & \(\approx 11.1\) \\
\(\infty\) & \(1\)            & \(\infty\)       \\
\bottomrule
\end{tabular}
\label{tab:alpha_crit_logistic}
\end{table}

\noindent\textit{Interpretation.} Higher steepness $\lambda$ \textbf{widens} the bistable range by
lowering the minimal \(\alpha_{\mathrm{crit}}\) (less amplification required to enter
bistability) while simultaneously \textbf{raising} the maximal \(\alpha_{\mathrm{crit}}\)
(stronger feedback required to annihilate the low state). As \(\lambda \to \infty\), the lower
threshold approaches~1 and the upper threshold diverges, so the low-expression state becomes
progressively harder to destroy. This is analogous to higher cooperativity in Hill functions,
which facilitates the onset of bistability while extending the amplification range over which
both states coexist.

The saddle-node structure summarised in
Table~\ref{tab:alpha_crit_logistic} can be sharpened into an explicit
characterisation that is convenient for parameter-estimation and
control-design purposes.

\begin{theorem}[Saddle-node characterisation and asymptotic widening of the bistable range]
\label{thm:bistability}
Let $\lambda,\theta>0$ and consider the steady-state condition
$x=\alpha f(x)$ for the logistic positive-autoregulation
model~\eqref{eq:autoreg_logistic}, where $f(x)=1/(1+e^{-\lambda(x-\theta)})$.
\begin{enumerate}
\item[(i)] (\emph{Saddle-node set.})  Saddle-node bifurcations occur exactly at
the parameter pairs $(\alpha,x)$ with $z=\lambda(x-\theta)$ satisfying the transcendental equation
\begin{equation}
  e^{z} - z \;=\; \lambda\theta - 1,
  \qquad
  \alpha_{\mathrm{crit}}(z) \;=\; \frac{(1+e^{z})(1+e^{-z})}{\lambda} \;=\; \frac{2+e^{z}+e^{-z}}{\lambda}.
  \label{eq:saddle_node_canonical}
\end{equation}
\item[(ii)] (\emph{Number of saddle-nodes.})  Equation~\eqref{eq:saddle_node_canonical} has
\emph{(a)} no real root if $\lambda\theta<2$, \emph{(b)} a single double root $z=0$ at
$\lambda\theta=2$ (degenerate cusp at $\alpha=2$), and \emph{(c)} exactly two real roots
$z_{-}<0<z_{+}$ for every $\lambda\theta>2$.  Bistability occurs precisely on the open
interval
\begin{equation}
  \alpha_{\mathrm{lower}} \;=\; \alpha_{\mathrm{crit}}(z_{+})
  \;<\; \alpha
  \;<\; \alpha_{\mathrm{crit}}(z_{-}) \;=\; \alpha_{\mathrm{upper}}.
\end{equation}
\item[(iii)] (\emph{Asymptotic widening as $\lambda\to\infty$ with $\theta$ fixed.})
The endpoints satisfy
\begin{equation}
  \alpha_{\mathrm{lower}} \;\longrightarrow\; \theta,
  \qquad
  \alpha_{\mathrm{upper}} \;\sim\; \frac{e^{\lambda\theta-1}}{\lambda}
  \;\longrightarrow\;\infty,
  \label{eq:bistability_asymptotics}
\end{equation}
so the bistable interval $(\alpha_{\mathrm{lower}},\alpha_{\mathrm{upper}})$ unboundedly
widens, recovering the step-function limit $(\theta,\infty)$.
\end{enumerate}
\end{theorem}

\begin{proof}
\emph{(i)} The simultaneous tangency conditions $x=\alpha f(x)$, $1=\alpha f'(x)$ together with
$f'=\lambda f(1-f)$ give $1/x=\lambda(1-f(x))$, i.e.\ $\lambda x(1-y)=1$ where $y=f(x)$.
Substituting $x=\theta+z/\lambda$ and $y=1/(1+e^{-z})$ rearranges to
$\lambda\theta+z=1+e^{z}$, i.e.\ $e^{z}-z=\lambda\theta-1$.  The corresponding $\alpha$ value
is $\alpha_{\mathrm{crit}}=x/y=1/(\lambda y(1-y))$, which simplifies via
$y(1-y)=e^{-z}/(1+e^{-z})^{2}$ to $\alpha_{\mathrm{crit}}=(1+e^{z})(1+e^{-z})/\lambda$.

\emph{(ii)} The function $g(z)=e^{z}-z$ has $g'(z)=e^{z}-1$, so $g$ is strictly convex with a
unique minimum $g(0)=1$.  The equation $g(z)=\lambda\theta-1$ has zero, one, or two real roots
depending on whether $\lambda\theta-1$ is below, equal to, or above $1$, i.e.\ on the sign of
$\lambda\theta-2$.  In the bistable case the two roots straddle $z=0$, and $\alpha_{\mathrm{crit}}$
is monotone in $|z|$ on each side of $z=0$, giving the ordering claimed.

\emph{(iii)} As $\lambda\to\infty$ with $\theta$ fixed,
$g(z)=\lambda\theta-1\to\infty$.  Using the closed form
$\alpha_{\mathrm{crit}}(z)=(e^{z}+2+e^{-z})/\lambda$ from~(i), it suffices
to track $e^{\pm z_{\pm}}$.  The positive root $z_{+}$ satisfies
$e^{z_{+}}=\lambda\theta-1+z_{+}$ with $z_{+}=\ln(\lambda\theta-1)+o(1)$;
hence $e^{-z_{+}}=O(1/\lambda)$ and
\[
\begin{aligned}
\alpha_{\mathrm{lower}}
&\;=\;
\frac{e^{z_{+}}+2+e^{-z_{+}}}{\lambda}
\;=\;
\theta + \frac{1+z_{+}+e^{-z_{+}}}{\lambda}\\[2pt]
&\;=\;
\theta + \frac{1+\ln(\lambda\theta-1)}{\lambda}+o\!\left(\tfrac{1}{\lambda}\right)
\;\longrightarrow\;\theta.
\end{aligned}
\]
The negative root $z_{-}$ satisfies $e^{z_{-}}-z_{-}=\lambda\theta-1$ with
$e^{z_{-}}\to 0$, so $-z_{-}=\lambda\theta-1-e^{z_{-}}$ and
$e^{-z_{-}}\sim e^{\lambda\theta-1}$.  Therefore
\[
\alpha_{\mathrm{upper}}
\;=\;
\frac{e^{z_{-}}+2+e^{-z_{-}}}{\lambda}
\;\sim\;
\frac{e^{\lambda\theta-1}}{\lambda}
\;\longrightarrow\;\infty.\qedhere\]
\end{proof}

The asymptotic estimates~\eqref{eq:bistability_asymptotics} are
quantitatively informative. The leading-order expansion of part~(iii) gives
$\alpha_{\mathrm{lower}}\approx\theta+(1+\ln(\lambda\theta-1))/\lambda$ and
$\alpha_{\mathrm{upper}}\approx e^{\lambda\theta-1}/\lambda$. At $\lambda=5$,
$\theta=1$ these predict $\alpha_{\mathrm{lower}}\approx 1+(1+\ln 4)/5\approx 1.48$
and $\alpha_{\mathrm{upper}}\approx e^{4}/5\approx 10.9$, in close agreement
with the exact values $1.58$ and $11.1$ in
Table~\ref{tab:alpha_crit_logistic}.  At $\lambda=10$, the same expansions
give $\alpha_{\mathrm{lower}}\approx 1.32$ and
$\alpha_{\mathrm{upper}}\approx 810$, again matching the exact roots to
leading order. The lower threshold approaches its limit $\theta$ through a
slow $(1+\ln\lambda)/\lambda$ correction, while the upper threshold diverges
exponentially in $\lambda\theta$. In control-design terms, the parameter
$\lambda$ acts as a single sigmoidicity knob that slides the low/high state
boundaries apart, exponentially on the upper side and only logarithmically
on the lower.

Theorem~\ref{thm:bistability} characterises the saddle-node \emph{set} but
does not classify the stability of the equilibria themselves; the next
corollary completes the bifurcation picture for the two-dimensional
mRNA--protein system~\eqref{eq:autoreg_logistic} by combining the
saddle-node analysis with a Jacobian computation and the Bendixson
criterion.

\begin{corollary}[Stability classification and global behaviour for the autoregulation model]
\label{cor:autoreg_stability}
Let $k_m,k_{dm},k_p,k_{dp},\lambda,\theta>0$ and consider the
mRNA--protein system~\eqref{eq:autoreg_logistic} on $\mathbb{R}^{2}$, with
$\alpha=k_m k_p/(k_{dm} k_{dp})$ and $f(x)=1/(1+e^{-\lambda(x-\theta)})$.
Let $\alpha_{\mathrm{lower}}$ and $\alpha_{\mathrm{upper}}$ be as in
Theorem~\ref{thm:bistability} (defined when $\lambda\theta>2$).  Then:
\begin{enumerate}
\item[(i)] (\emph{Forward invariance and absence of closed orbits.}) The
closed rectangle
$\mathcal{B}_{\mathrm{auto}}=[0,k_m/k_{dm}]\times[0,\alpha]$ is forward
invariant; every trajectory of~\eqref{eq:autoreg_logistic} on $\mathbb{R}^{2}$
enters $\mathcal{B}_{\mathrm{auto}}$ in finite time.  The divergence of the
right-hand side is $-(k_{dm}+k_{dp})<0$ throughout $\mathbb{R}^{2}$, so by
Bendixson's negative criterion no closed orbit exists.
\item[(ii)] (\emph{Local stability classification.}) An equilibrium
$(m^{*},x^{*})$ with $x^{*}=\alpha f(x^{*})$ is locally asymptotically
stable iff $\alpha f'(x^{*})<1$, and is a saddle iff $\alpha f'(x^{*})>1$;
saddle-node bifurcations correspond to $\alpha f'(x^{*})=1$, in agreement
with Theorem~\ref{thm:bistability}.
\item[(iii)] (\emph{Global behaviour.})
\begin{itemize}
\item If $\lambda\theta\leq 2$, or $\lambda\theta>2$ and
  $\alpha\notin[\alpha_{\mathrm{lower}},\alpha_{\mathrm{upper}}]$, then the
  unique equilibrium of~\eqref{eq:autoreg_logistic} is globally
  asymptotically stable on $\mathbb{R}^{2}$.
\item If $\lambda\theta>2$ and
  $\alpha\in(\alpha_{\mathrm{lower}},\alpha_{\mathrm{upper}})$, then there
  are exactly three equilibria; the lowest and highest are locally
  asymptotically stable and the middle is a saddle, and
  $\mathbb{R}^{2}_{\geq 0}$ minus the saddle is partitioned into the basins
  of the two stable equilibria by the one-dimensional stable manifold of
  the saddle (the separatrix).
\end{itemize}
\end{enumerate}
\end{corollary}

\begin{proof}
\emph{(i)} On the boundary $\{m=0\}$, $\dot m = k_m f(x) > 0$; on
$\{m=k_m/k_{dm}\}$, $\dot m = k_m(f(x)-1)\leq 0$.  On $\{x=0\}$,
$\dot x = k_p m \geq 0$; on $\{x=\alpha\}$ and $m\leq k_m/k_{dm}$,
$\dot x = k_p m - k_{dp}\alpha \leq k_p(k_m/k_{dm}) - k_{dp}\alpha = 0$.
By Nagumo's theorem~\citep{blanchini2008set}, $\mathcal{B}_{\mathrm{auto}}$
is forward invariant; outside $\mathcal{B}_{\mathrm{auto}}$ similar
inequalities give monotone convergence into the box.  The divergence is
$\partial \dot m/\partial m + \partial \dot x/\partial x = -k_{dm}-k_{dp}<0$
identically, so Bendixson's criterion rules out closed orbits in any
simply connected region.

\emph{(ii)} The Jacobian at $(m^{*},x^{*})$ is
$J=\bigl(\begin{smallmatrix}-k_{dm}&k_m f'(x^{*})\\k_p&-k_{dp}\end{smallmatrix}\bigr)$,
with $\operatorname{tr}J=-(k_{dm}+k_{dp})<0$ and
$\det J = k_{dm} k_{dp}\bigl(1-\alpha f'(x^{*})\bigr)$.  Stability holds iff
$\det J>0$, i.e.\ iff $\alpha f'(x^{*})<1$; the equilibrium is a saddle iff
$\det J<0$, i.e.\ $\alpha f'(x^{*})>1$.  At a saddle-node, $\det J=0$
gives $\alpha f'(x^{*})=1$, identical to the tangency condition derived
in~\eqref{eq:tangency_logistic}.

\emph{(iii)} For each parameter regime, Theorem~\ref{thm:bistability}
gives the number of fixed points.  In the monostable regimes the unique
equilibrium has $\alpha f'(x^{*})<1$ (since the function $\alpha f$
crosses the diagonal transversally with slope less than $1$ when there
is no second crossing), hence is locally asymptotically stable by~(ii);
combined with forward invariance from~(i) and the Bendixson criterion,
Poincar\'e--Bendixson forces the $\omega$-limit of every trajectory to
be the unique equilibrium, giving global asymptotic stability.  In the
bistable regime, the three roots of $x=\alpha f(x)$ correspond
geometrically to the diagonal crossing $\alpha f$ in slow--fast--slow
order, so the slopes alternate $<1,\,>1,\,<1$, classifying the outer two
as stable and the middle as a saddle.  Bendixson again rules out closed
orbits, so by Poincar\'e--Bendixson every trajectory converges to one of
the equilibria; the basin boundary is the one-dimensional stable
manifold of the saddle by standard hyperbolic-saddle theory.
\end{proof}

Corollary~\ref{cor:autoreg_stability} sharpens Theorem~\ref{thm:bistability}
in two practically important ways. First, it identifies which fixed points
are stable in each parameter regime, so simulations need not infer
stability from initial-condition sweeps. Second, it certifies global
asymptotic stability in the monostable regimes, including the
high-amplification regime $\alpha>\alpha_{\mathrm{upper}}$ where the basal
escape mechanism of Section~\ref{ex:positive_autoregulation} operates: at
$\alpha\approx 600\gg\alpha_{\mathrm{upper}}\approx 2.82$, every trajectory
converges to the unique high-expression equilibrium, so the
$\sim$44~minute escape time is not a metastable transient but the
deterministic approach to a globally attracting state.

The logistic model's key advantage is most pronounced in scenarios with low protein expression.
For the Hill function, the production rate at \(x = 0\) is exactly zero. Consequently, when
protein level reaches zero, mRNA synthesis halts: \(\dot{m}|_{x=0} = -k_{dm}m\), leading to
exponential mRNA decay \(m(t) = m(0)e^{-k_{dm}t}\), and subsequently protein decay
\(x(t) = x(0)e^{-k_{dp}t}\). The system spirals irreversibly toward the absorbing state
\((m,x) = (0,0)\). Escape requires the protein level to reach the unstable separatrix at
\(x^* \approx 0.041\), an event with vanishingly small probability in a noisy low-copy-number
environment~\citep{lipshtat2006genetic}.

The logistic function maintains a non-zero production rate even at \(x = 0\):
\begin{equation}
f(0) = \frac{1}{1 + e^{\lambda\theta}}.
\end{equation}
For \(\lambda = 3\) and \(\theta = 1\):
\(e^{\lambda\theta} = e^3 \approx 20.086\), so
\(f(0) \approx 1/21.086 \approx 0.0474\),
giving a basal mRNA synthesis rate of
\(k_m f(0) \approx (0.003)(0.0474) \approx 0.000142\)~molecules\,s\(^{-1}\).
At quasi-steady state (with protein held at zero), the basal mRNA level is:
\begin{equation}
m_{\text{basal}}
= \frac{k_m f(0)}{k_{dm}}
= \frac{0.000142}{0.001}
\approx 0.142\ \text{molecules}.
\end{equation}
While fractional molecule numbers are non-physical deterministically, they represent a
time-averaged stochastic occupancy of approximately 14\% of the time with one mRNA molecule
present.

Once the basal mRNA pool reaches its quasi-steady-state value $m_{\mathrm{basal}}$
(on the timescale $1/k_{dm} \approx 16.7$~min), the resulting
protein accumulation rate is
\begin{equation}
\begin{aligned}
\dot{x}\big|_{\mathrm{qss}}
&\;\approx\; k_p\,m_{\mathrm{basal}} - k_{dp}\,x
\;\approx\; (0.002)(0.142) - (0.00001)(0.01)\\
&\;\approx\; 2.84 \times 10^{-4}~\text{molecules\,s}^{-1}.
\end{aligned}
\end{equation}
(The instantaneous initial rate $\dot{x}(0) = k_p m(0) - k_{dp} x(0) \approx 2 \times 10^{-5}$~molecules s$^{-1}$
is roughly $14\times$ smaller, because $m(0)=0.01 \ll m_{\mathrm{basal}}=0.142$;
the quasi-steady-state rate above governs the dynamics once the mRNA pool has
equilibrated.)
To reach \(x \approx 1\) from \(x \approx 0.01\) requires accumulating approximately 0.99
protein units, giving a crude linear estimate:
\begin{equation}
t_{\mathrm{escape}}
\sim \frac{0.99}{2.84 \times 10^{-4}}
\approx 3{,}480\ \mathrm{s}
\approx 58\ \mathrm{minutes}.
\end{equation}
Numerical simulation shows escape in approximately 2650~seconds (\(\sim\)44~minutes),
somewhat faster than this estimate. The discrepancy arises because the crude estimate treats
the accumulation rate as constant, whereas in reality the positive feedback accelerates protein
production as \(x\) approaches \(\theta = 1\), shortening the escape time. The
order-of-magnitude agreement confirms the dominant role of basal production in driving escape.

This escape mechanism mirrors observed behaviour in the \textit{gal} operon, where leaky basal
expression prevents complete transcriptional shutdown during nutrient shifts, enabling rapid
induction upon re-exposure to galactose~\citep{hua1974multiple,weickert1993galactose}.
Experimental studies have shown that small leaky production, on the order of 1 to 5\% of
maximal expression, is sufficient to maintain responsiveness in bistable systems subjected to
environmental fluctuations~\citep{weickert1993galactose}.

To rigorously test the theoretical predictions, we numerically integrate
Equations~\eqref{eq:autoreg_logistic} and \eqref{eq:autoreg_hill} over the time interval
\([0,\,10{,}000]\)~seconds, starting from \(m(0) = 0.01\) and \(x(0) = 0.01\).
Figure~\ref{fig:autoreg_sim} presents the protein trajectories \(x(t)\) for both models.

For the \textbf{logistic model}: the protein level initially remains very low as basal mRNA
production slowly accumulates. At approximately \(t \approx 2650\)~seconds (\(\sim\)44~minutes),
the protein level surpasses the activation threshold \(\theta = 1\), triggering strong positive
feedback. The system rapidly accelerates, reaching \(x \approx 38\) at \(t = 10{,}000\)~s and
approaching asymptotically the high steady state \(x_{\mathrm{ss}} \approx \alpha = 600\),
consistent with monostability at the high state for \(\alpha \gg \alpha_{\mathrm{upper}}\)
(see Theorem~\ref{thm:bistability}; here $\alpha_{\mathrm{upper}}\approx 2.82$ for $\lambda=3,\theta=1$).

For the \textbf{Hill model}: mRNA synthesis is negligible throughout the simulation because
the Hill activation at the initial protein level satisfies
\begin{equation}
h(0.01) = \frac{(0.01)^3}{(0.01)^3 + 1}
= \frac{10^{-6}}{1 + 10^{-6}}
\approx 1 \times 10^{-6},
\end{equation}
making the mRNA production rate \(k_m h(x) \approx 3 \times 10^{-9}\)~molecules\,s\(^{-1}\) more than three orders
of magnitude smaller than the mRNA degradation rate \(k_{dm}\,m(0) = 10^{-5}\)~molecules\,s\(^{-1}\).
Consequently, mRNA decays exponentially as \(m(t) \approx m(0)\,e^{-k_{dm}t}\), and the
protein satisfies the resulting linear ODE \(\dot{x} = k_p m(t) - k_{dp} x\) with explicit
solution
\begin{equation}
x(t) = x(0)\,e^{-k_{dp}t}
+ \frac{k_p\,m(0)}{k_{dm} - k_{dp}}
\!\left(e^{-k_{dp}t} - e^{-k_{dm}t}\right).
\label{eq:hill_explicit}
\end{equation}
Substituting the parameter values gives
\(x(t) \approx 0.0302\,e^{-10^{-5}t} - 0.0202\,e^{-0.001t}\),
which \emph{rises} from \(x(0) = 0.01\) to a transient peak of
\(x_{\mathrm{peak}} \approx 0.029\) at \(t \approx 4{,}245\)~s (\(\approx 71\)~min)
according to this linear analytical approximation. The actual nonlinear simulation
places the peak slightly later (near $t \approx 4{,}900$~s) due to the small but
non-zero Hill production at $x \neq 0$, but the peak height $x_{\mathrm{peak}} \approx 0.029$
is essentially identical. After this peak, the protein decays negligibly at the
very slow rate \(k_{dp} = 10^{-5}\)~s\(^{-1}\) (half-life \(\approx 19\) hours); at
\(t = 10{,}000\)~s the protein level is \(x \approx 0.028\). This transient peak is
\emph{not} a fixed point: one can verify directly that
\(600\,h(0.028) \approx 0.013 \neq 0.028\), confirming that no balance between production
and degradation exists at this value. Crucially, the trajectory remains at all times below
the unstable separatrix at
\begin{equation}
x^* \approx \frac{1}{\sqrt{600}} \approx 0.041,
\end{equation}
obtained from \(x^* = 600\,h(x^*)\) which for small \(x^*\) approximates as \(x^* \approx 600\,(x^*)^3\) (since \(h(x^*)\approx (x^*)^3\) when \((x^*)^3\ll 1\)), giving \((x^*)^2 \approx 1/600\).
Since the trajectory never reaches this separatrix and no basal mRNA production exists to
sustain protein accumulation, no escape is possible through intrinsic dynamics alone---a
statistically improbable event in a low-copy-number cellular environment.

\begin{figure}[t]
\centering
\includegraphics[height=6cm, width=8cm]{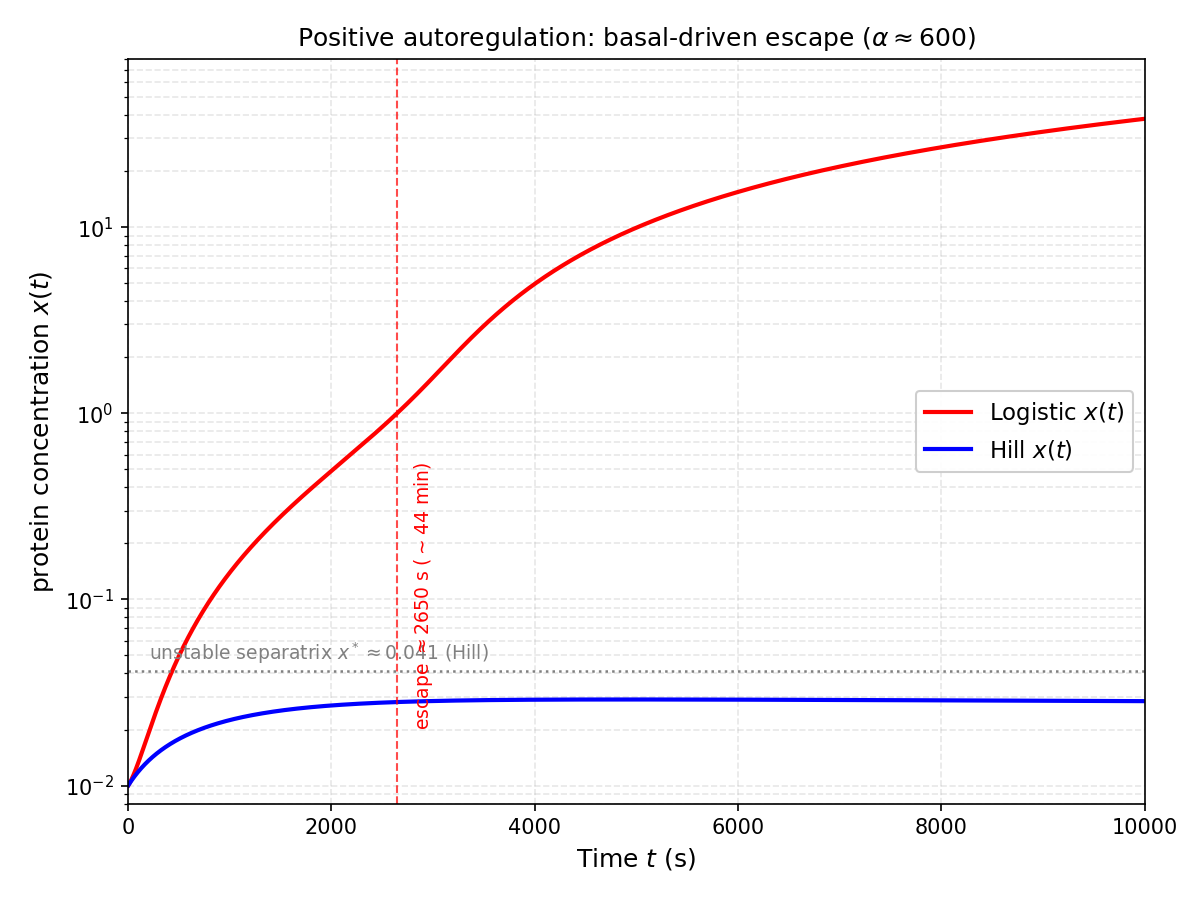}
\caption{Protein dynamics in positive autoregulation under low initial conditions
(\(m(0) = 0.01\), \(x(0) = 0.01\)) with feedback amplification \(\alpha \approx 600\).
\textbf{Logistic model (monostable-high regime):} The system escapes the off-state in
approximately 2650~s (\(\sim\)44~min) due to basal production
\(k_m f(0) \approx 0.000142\)~molecules\,s\(^{-1}\), reaching \(x \approx 38\) at 10{,}000~s and
approaching \(x_{\mathrm{ss}} \approx 600\) asymptotically.
\textbf{Hill model (bistable regime):} mRNA decays exponentially and protein reaches a
transient peak near \(x \approx 0.029\) (analytically predicted at \(t \approx 4{,}245\)~s) before decaying
negligibly; at \(t = 10{,}000\)~s the protein remains near \(x \approx 0.028\), well
below the unstable separatrix at \(x^* \approx 0.041\). Because \(h(x) \approx x^3 \approx 0\)
throughout, no basal production exists to drive escape. This illustrates the critical role
of basal expression in maintaining cellular responsiveness.}
\label{fig:autoreg_sim}
\end{figure}

While the simulations employ \(k_{dp} = 10^{-5}\)~s\(^{-1}\), representing the upper range of
transcription factor stability, the qualitative distinction between models persists across
physiologically plausible degradation rates. This demonstrates that the basal expression
mechanism is not an artifact of extreme parameter choices but a fundamental feature of logistic
formulations that accurately captures regulatory dynamics in the moderate-stability regime
characteristic of autoregulatory transcription factors~\citep{belle2006quantification,
weickert1993galactose}.

The simulation outcomes align with experimental observations across several domains. Becskei and Serrano~\cite{becskei2000engineering} constructed positive feedback loops in \textit{E.~coli} and observed that basal leakage prevents trapping in off-states and stabilises expression under noise; circuits without leakage exhibited hysteresis and irreversible commitment to low-expression states. The \textit{lac} operon exhibits basal expression of approximately $1$--$2\%$ of maximal that enables rapid induction upon lactose exposure even after prolonged growth in glucose~\citep{setty2003detailed,ozbudak2004multistability}, and the \textit{gal} operon similarly maintains low-level expression preventing transcriptional shutdown~\citep{weickert1993galactose}. Single-cell studies directly visualised stochastic transitions between expression states in bistable circuits on timescales of minutes to hours~\citep{ozbudak2004multistability}, consistent with our $\sim$44-minute simulation. Hill-type models without explicit leakage fail to predict dynamics in circuits with engineered rapid degradation (ssrA tags), requiring ad~hoc basal terms~\citep{becskei2000engineering}; the logistic model incorporates this naturally.

Positive autoregulation exemplifies the rich interplay between feedback architecture, molecular
noise, and expression dynamics. Hill functions excel at describing cooperative switches but fail
at low expression due to zero basal rate, trapping systems in off-states. Logistic functions,
with inherent non-zero basal production, naturally capture noise-driven escape without ad hoc
parameters. Our analysis, grounded in experimentally derived \textit{E.~coli} parameters and
validated through numerical simulation, demonstrates that basal expression is a fundamental
feature enabling cellular responsiveness: the logistic model's escape in \(\sim\)44~minutes,
driven solely by basal production of \(\sim 0.000142\)~molecules\,s\(^{-1}\), aligns quantitatively with
observations in the \textit{gal} operon and related systems.

\subsection{Basal Expression in a Boolean-Derived Network}
\label{subsec:minimal_boolean}

The preceding two examples examined, in a single autoregulating gene and in a
two-gene oscillator, how the vanishing of the Hill kernel at zero input removes
basal expression and drives the model toward shutdown. We close this section
with a small \emph{Boolean-derived} network, which exhibits the same
basal-expression contrast at the level of a network rather than a single
motif, in a system small enough to be displayed in full and checked in closed
form. It also serves as a transparent preview of the $80$-gene Boolean-derived
experiment of Section~\ref{sec:numerical_comparison}.

Consider six genes $g_1,\dots,g_6$ governed by the Boolean update rules
\begin{equation}
\begin{aligned}
g_1 &\leftarrow \textsc{False}, & \qquad g_4 &\leftarrow \lnot g_3,\\
g_2 &\leftarrow g_1,            & \qquad g_5 &\leftarrow g_3 \vee g_6,\\
g_3 &\leftarrow \lnot g_4,      & \qquad g_6 &\leftarrow g_3 \wedge \lnot g_2.
\end{aligned}
\label{eq:boolean_motif_rules}
\end{equation}
These rules exercise each elementary construct of the De~Morgan translation of
Section~\ref{sec:prelim}: a constant ($g_1$), a single positive literal
($g_2$), a pair of mutually repressing genes forming a sub-toggle
($g_3,g_4$), a disjunction ($g_5$), and a conjunction of a positive and a
negative literal ($g_6$). Sending each literal to a kernel, each conjunction to
a product, and the disjunction through the De~Morgan product
formula~\eqref{eq:demorgan_0} gives the regulatory functions
\begin{equation}
\begin{aligned}
\Phi_1 &= 0, & \qquad \Phi_4 &= s^-(x_3),\\
\Phi_2 &= s^+(x_1), & \qquad \Phi_5 &= 1-\bigl(1-s^+(x_3)\bigr)\bigl(1-s^+(x_6)\bigr),\\
\Phi_3 &= s^-(x_4), & \qquad \Phi_6 &= s^+(x_3)\,s^-(x_2),
\end{aligned}
\label{eq:boolean_motif_phi}
\end{equation}
where $s^\pm$ denotes the chosen kernel pair: the logistic kernels $f^\pm$, or
the Hill kernels $h^\pm$ evaluated at the domain-guarded argument
$\max(x_j,0)$. Guarding the Hill argument---the patch examined in
Section~\ref{sec:general_failure}---keeps the kernel real-valued under any
numerical excursion, so that the comparison below isolates the \emph{modelling}
behaviour of the two kernels from the numerical pathologies of
Section~\ref{sec:hill_pathologies_empirical}. Each gene then obeys
$\dot x_i = \kappa_i\Phi_i(\mathbf{x})-\gamma_i x_i$ as in~\eqref{eq:ode_0}; we
integrate from the initial state $\mathbf{x}(0)=(0.8,0.6,0.9,0.2,0.5,0.4)$ with
$\kappa=(1,1,1.2,1.2,1,0.9)$, $\gamma=(1,\dots,1)$, threshold $\theta=0.5$, and
the non-integer exponent $n=3.50918$, the logistic
steepness matched by $\lambda=n/\theta$.

The decisive gene is $g_2$. Its sole activator is $g_1$, whose Boolean rule is
the constant \textsc{False}; hence $\Phi_1=0$ and $g_1$ decays exponentially,
$x_1(t)=x_1(0)e^{-\gamma_1 t}\to0$. The steady level of $g_2$ is then governed
entirely by the value of its kernel at zero input. Under the Hill kernel
$h^+(0;\theta,n)=0$, so once $x_1$ has decayed the production term
$\kappa_2 h^+(x_1)$ vanishes and $g_2$ converges to the steady state
\[
  x_2^{\ast}\big|_{\mathrm{Hill}}
  \;=\; \frac{\kappa_2\,h^+(0)}{\gamma_2} \;=\; 0,
\]
exactly zero, with no basal expression whatsoever. Under the logistic kernel
the basal activation is strictly positive,
\begin{equation}
f^+(0;\theta,\lambda)=\frac{1}{1+e^{\lambda\theta}}=\frac{1}{1+e^{n}}\approx0.029052,
\end{equation}
the value depending only on $n$, since the matching $\lambda=n/\theta$ makes
$\lambda\theta=n$. Gene $g_2$ therefore converges to the strictly positive
basal steady state
\begin{equation}
x_2^{\ast}\big|_{\mathrm{logistic}}=\frac{\kappa_2\,f^+(0)}{\gamma_2}
        =\frac{1\cdot0.029052}{1}=0.029052,
\label{eq:boolean_motif_basal}
\end{equation}
in exact agreement with the simulated value. The biological reading is the one
developed throughout this section: promoters are never completely silent, so
the logistic level $x_2^{\ast}\approx0.029$---about $3\%$ of the gene's maximal
expression $\kappa_2/\gamma_2$---is a faithful basal level, whereas the Hill
prediction of identically zero expression is not. Because $x_2$ also enters the
production rule of $g_6$ through the literal $\lnot g_2$, the two models differ
slightly on the downstream genes as well; but the cleanest and exactly
predictable contrast is the one on $g_2$ itself.

Integrating to $t=200$, the horizon of the $80$-gene experiment, and continuing
to $t=2000$ shows that the contrast is permanent rather than transient. The
genes $g_3,\dots,g_6$ settle to bounded steady states within $t\approx15$;
thereafter the logistic value $x_2=0.029052$ holds unchanged, while the Hill
value of $g_2$ decays until it is numerically indistinguishable from zero. The
Hill model does not place $g_2$ at a small basal level---it sends $g_2$ to its
true steady state of zero, so the separation between the two models on $g_2$
does not close. Figure~\ref{fig:boolean_motif} shows the two integrations and
the $g_2$ off-state. Both steady states---$0$ for the Hill model and
$\kappa_2 f^+(0)/\gamma_2$ for the logistic---are equilibria of the respective
vector fields, so the contrast is a property of the models themselves,
reproduced across independent re-implementations and solver tolerances, and not
an artefact of the integrator.

This example isolates one of the two faces of the Hill kernel's degeneracy at
the boundary $x=0$ of the non-negative orthant. The face seen here is a
\emph{modelling} one: because $h^+(0)=0$, every gene whose activators are all
inactive is assigned production exactly zero---the same loss of basal
expression that, in the bistable autoregulation circuit of
Section~\ref{ex:positive_autoregulation}, produced the absorbing off-state. The
other face is \emph{numerical} and is the subject of
Section~\ref{sec:numerical_comparison}: at and just outside that same boundary
the Hill field forfeits its smoothness and, for negative arguments, its
real-valuedness. The two are not independent defects but consequences of the
single fact that the Hill kernel degenerates at $x=0$; the domain guard adopted
above suppresses the numerical face, and the persistence of the zero steady
state of $g_2$ shows that it leaves the modelling face untouched. The example
is instructive precisely because the Hill failure here is entirely quiet: the
trajectories are smooth and bounded and the integration reports no difficulty,
yet the model has converged to a biologically impossible state---a Hill
simulation that ``looks correct'' can still be wrong. The logistic steady
state, by contrast, is available in closed form,
Equation~\eqref{eq:boolean_motif_basal}, giving an exact analytic check that an
$80$-gene black-box integration cannot provide.

\begin{figure}[t]
  \centering
  \includegraphics[width=\textwidth]{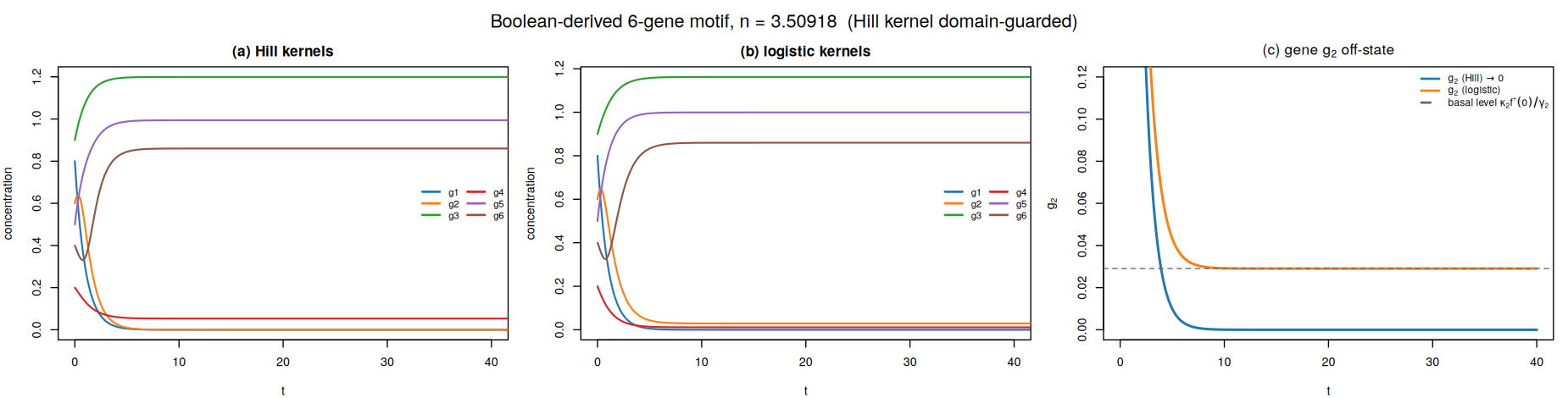}
  \caption{%
    Minimal Boolean-derived $6$-gene motif of
    Equation~\eqref{eq:boolean_motif_rules}, integrated under the Hill and
    logistic kernels with the non-integer exponent $n=3.50918$ of the
    $80$-gene experiment; the Hill kernel is evaluated at the domain-guarded
    argument $\max(x,0)$. \textbf{(a)},~\textbf{(b)}~Trajectories of all six
    genes over $t\in[0,40]$ under the Hill and the logistic kernels; the genes
    $g_3,\dots,g_6$ settle by $t\approx15$. \textbf{(c)}~The off-state of gene
    $g_2$, whose sole activator $g_1$ is constitutively off, with the vertical
    axis restricted to $[0,0.12]$ to resolve the two steady levels (both genes
    start at $x_2=0.6$, above the panel). Under the Hill kernel $g_2$ converges
    to its steady state of exactly zero; under the logistic kernel it converges
    to the basal level $\kappa_2 f^+(0)/\gamma_2=0.029052$ (grey dashed line),
    the closed-form prediction of Equation~\eqref{eq:boolean_motif_basal}.
    Parameters: $\kappa=(1,1,1.2,1.2,1,0.9)$, $\gamma=(1,\dots,1)$,
    $\theta=0.5$, $\lambda=n/\theta$,
    $\mathbf{x}(0)=(0.8,0.6,0.9,0.2,0.5,0.4)$.%
  }
  \label{fig:boolean_motif}
\end{figure}

\section{Numerical Integration of Boolean-Derived ODE Systems: Hill Functions with Real Exponent versus Logistic Functions}
\label{sec:numerical_comparison}

To compare the numerical behaviour of Hill-based and logistic-based ODE systems
at scale, we designed a simulation protocol implemented in \textit{Mathematica}
that proceeds through four successive stages:
(i)~construction of a Boolean regulatory network;
(ii)~automatic translation of each Boolean update rule into a continuous ODE;
(iii)~numerical integration of the resulting high-dimensional system; and
(iv)~extraction of state snapshots at prescribed observation times.
The same Boolean network and the same parameter sets are used in both the
Hill-function (notebook~\texttt{S1}) and the logistic-function
(notebook~\texttt{S2}) experiments, so that any difference in the simulated
trajectories is attributable solely to the choice of regulatory function.

\subsection{Experimental Protocol}
\label{subsec:protocol}

\subsubsection{Boolean Network Construction}
The Boolean network $\mathcal{F}$ consists of $N = 80$ variables
$\mathbf{x} = (x_1, x_2, \ldots, x_{80})$, each governed by a propositional
update rule of the form
\[
  x_i \;\leftarrow\; \varphi_i(x_1,\ldots,x_{80}),
  \qquad i = 1,\ldots,80,
\]
where each $\varphi_i$ is a Boolean formula over the variables and their
negations, following the Boolean network formalism introduced by
Kauffman~\cite{kauffman1969metabolic} and applied to gene regulatory systems by
Albert and Othmer~\cite{albert2003topology}. The network exhibits a wide
spectrum of regulatory complexity. At one extreme, seven variables are assigned
constant rules: four are fixed to \textsc{False}
($x_{24}, x_{37}, x_{53}, x_{77}$), meaning they receive no production input
and undergo pure exponential decay, and three are fixed to \textsc{True}
($x_{61}, x_{67}, x_{76}$), meaning they are constitutively expressed at their
maximal rate $\kappa_i$. At the other extreme, variables such as $x_1$, $x_3$,
and $x_4$ are regulated by 15, 21, and 22 distinct conjunctive clauses
respectively, each clause encoding a specific combination of activating and
repressing signals. Between these extremes lie simple two-literal rules such as
$x_2 \leftarrow \lnot x_3 \wedge \lnot x_{42}$ or
$x_7 \leftarrow \lnot x_{17} \wedge \lnot x_{35}$, which represent
straightforward mutual repression gates.

Prior to conversion, each formula $\varphi_i$ is reduced to a canonical minimal
disjunctive normal form (DNF) using \texttt{BooleanMinimize} in \textsc{DNF}
mode. This preprocessing step eliminates redundant and duplicate conjunctive
clauses, such as absorbed terms of the form $C \vee (C \wedge D) \to C$, or
tautologies $C \vee \lnot C \to \textsc{True}$, and propagates constants arising
from the seven fixed variables. Without this step, duplicate clauses surviving
into the ODE translation would introduce artificial integer coefficients
($2\times$, $3\times$, \ldots) in the production terms, inflating $\Phi_i$
beyond its biologically meaningful $[0,1]$ bound.

\subsubsection{Translation to ODEs}
The minimised Boolean network is converted automatically into a continuous ODE system following the De~Morgan formalism developed in the companion paper~\citep{belgacem2026framework} and recalled in Section~\ref{sec:prelim}: each positive literal $x_j$ becomes $h^+(x_j)$ or $f^+(x_j)$; each negative literal $\lnot x_j$ becomes $h^-(x_j) = 1 - h^+(x_j)$ or $f^-(x_j) = 1 - f^+(x_j)$; each conjunction is mapped to the product of the corresponding sigmoidal terms; and each disjunction is mapped via the recursive De~Morgan product formula~\eqref{eq:demorgan_0}. The explicit deployment of this formula for $m$-clause disjunctions within logistic-based ODE systems, combined with \texttt{BooleanMinimize} preprocessing, constitutes, to the best of our knowledge, a contribution not previously made explicit in the gene-regulatory-network modelling literature. By contrast, the weighted-sum formulation of Samuilik \emph{et al.}~\citep{samuilik2022mathematical} also preserves the $[0,1]$ bound but does so by collapsing all regulatory logic into a shared threshold and a single sigmoid, thereby preventing regulator-specific tuning and obscuring the distinction between AND and OR combinatorial gates (see the companion paper~\citep{belgacem2026framework} for a detailed comparison). For example, the 15-clause disjunctive rule for $x_1$ yields a product of 15 complementary terms via~\eqref{eq:demorgan_0}, while the single-clause rule for $x_2$ reduces directly to the product $f^-(x_3)\,f^-(x_{42})$.

The resulting ODE for variable $x_i$ takes the form~\eqref{eq:ode_0}, with $\Phi_i(\mathbf{x}) \in [0,1]$ the continuous approximation of $\varphi_i$. Variables fixed to \textsc{False} satisfy $\Phi_i \equiv 0$, giving $\dot{x}_i = -\gamma_i x_i$ (pure exponential decay to zero); those fixed to \textsc{True} satisfy $\Phi_i \equiv 1$ and converge monotonically to $x_i^* = \kappa_i/\gamma_i$. By construction, all variables satisfy $x_i(t) \leq \kappa_i/\gamma_i$ for all $t \geq 0$.

\subsubsection{Parameter Selection}
Each of the five parameter objects---production rates $\boldsymbol{\kappa}$,
degradation rates $\boldsymbol{\gamma}$, activation/repression thresholds
$\boldsymbol{\theta}$, the shared cooperativity coefficient $n$, and the initial
condition vector $\mathbf{x}(0)$---is first
drawn independently from a uniform distribution. This constitutes the
\emph{exploratory} draw intended to verify that the pipeline produces
biologically plausible ranges. To ensure that any difference between the Hill-based and logistic-based
simulations is caused solely by the choice of regulatory function and not by
parameter randomness, all five objects are subsequently \emph{overwritten} with
a single fixed realisation generated by the same random procedure in a prior
session. This two-phase design---first a random draw to explore the feasible
space, then a fixed assignment shared across experiments---is standard practice
in computational systems biology: it guarantees reproducibility while
demonstrating that the parameter values are not hand-tuned but drawn from
biologically motivated distributions.

The fixed value $n = 3.50918$ is particularly important: it lies strictly in
the open interval $(3, 4)$ and is therefore \emph{non-integer}. As detailed in the
companion paper~\citep{belgacem2026framework} and recalled below in Section~\ref{sec:hill_pathologies_empirical}, this is precisely the regime that exposes the
numerical pathologies of the Hill function. Table~\ref{tab:params} summarises
the distributions and the resulting fixed values.

\begin{table}[!htbp]
  \centering
  \caption{Parameter distributions and fixed values used in both ODE
           experiments. All 80-dimensional vectors share the same realisation
           across the Hill and logistic notebooks.}
  \label{tab:params}
  \begin{tabular}{llll}
    \hline
    Parameter & Symbol & Draw distribution & Range of fixed values \\
    \hline
    Production rate   & $\kappa_i$ & $\mathrm{U}(50,100)$  & $[53.2,\; 99.6]$   \\
    Degradation rate  & $\gamma_i$ & $\mathrm{U}(0.25,2)$  & $[0.28,\; 1.98]$   \\
    Threshold         & $\theta_i$ & $\mathrm{U}(10,20)$   & $[10.4,\; 19.8]$   \\
    Cooperativity     & $n$        & $\mathrm{U}(1,5)$     & $3.50918$ (fixed)  \\
    Initial condition & $x_i(0)$  & $\mathrm{U}(0,100)$   & $[0.47,\; 99.3]$   \\
    \hline
  \end{tabular}
\end{table}

\subsubsection{Numerical Integration and State Extraction}
The ODE system~\eqref{eq:ode_0} is integrated from $t = 0$ to $t = 200$ using \texttt{NDSolve} with the fixed initial conditions. Both the Hill and logistic versions call \texttt{NDSolve} with identical settings (default adaptive step-size control, default error tolerances); no solver-specific tuning is applied, so that any difference between the two simulations is attributable to the regulatory kernel alone, all other factors held fixed. We evaluate the numerical solution at the observation times $t^* \in \{10, 30, 50, 100, 150\}$, returning the association $\{x_i \mapsto x_i(t^*)\}_{i=1}^{80}$ for downstream analysis (e.g.\ attractor identification by fixed-point comparison).

\subsection{Definition of the Regulatory Functions}
\label{subsec:reg_functions_def}
Both experiments use the same functional signatures; only the kernel is changed.

\subsubsection{Hill Functions}

The standard Hill activation and repression functions with threshold $\theta$
and (real-valued) cooperativity exponent $n$ are
\begin{equation}
  h^+(x;\,\theta,n) \;=\; \frac{x^n}{\theta^n + x^n},
  \qquad
  h^-(x;\,\theta,n) \;=\; \frac{\theta^n}{\theta^n + x^n}.
  \label{eq:hill}
\end{equation}

\subsubsection{Logistic Functions}

With the steepness matching $\lambda = n/\theta$, the logistic activation and
repression functions are
\begin{equation}
  f^+(x;\,\theta,n) \;=\; \frac{1}{1+e^{-(n/\theta)(x-\theta)}},
  \qquad
  f^-(x;\,\theta,n) \;=\; \frac{1}{1+e^{+(n/\theta)(x-\theta)}}.
  \label{eq:logistic}
\end{equation}

\subsection{Empirical Confirmation of Hill-Function Pathologies}
\label{sec:hill_pathologies_empirical}

The theoretical sources of numerical instability in Hill-function ODE systems
were identified in the companion paper~\citep{belgacem2026framework}. We recall the three
interacting mechanisms for convenience, then show that each one is directly
observable in the output of the 80-variable simulation.

\subsubsection{Mechanism 1: Complex Arithmetic for Non-Integer $n$ and $x < 0$}
ODE solvers routinely produce small negative values for concentrations during
integration steps. For non-integer $n \notin \mathbb{N}$, the expression $x^n$
is undefined over the reals when $x < 0$; using the principal branch of the
complex power gives $x^n = |x|^n e^{i\pi n} \in \mathbb{C}$, immediately
corrupting the entire right-hand side of the ODE system.  Quantitatively,
the imaginary part of $x^n$ for $x<0$ equals $|x|^n\sin(\pi n)$, with
$|\sin(\pi n)|$ bounded away from zero for any non-integer $n$ and reaching
its maximum $1$ at half-integer values; for the value $n=3.50918$ used in
the present experiment, $\sin(\pi n)\approx -0.9996$, so even a rounding-level overshoot
$|x|\sim 10^{-15}$ gives $x^{n}$ an imaginary part of magnitude
$\sim 10^{-53}$, which then propagates through $h^{\pm}(x;\theta,n)$ and the
right-hand side and contaminates every subsequent integration step.  (For
odd-integer $n$ the surrogate is real-valued but $x^n<0$ when $x<0$, so
$h^+(x;\theta,n)<0$, which is biologically inadmissible for a normalised
activation probability.  This case is not encountered in the present
experiment, but applies to any Hill-based simulation that adopts an integer
cooperativity to evade the fractional-power issue.)
The complex-arithmetic pathway is therefore active here,
as confirmed by the solver diagnostic emitted at the first integration attempt:

\smallskip
\begin{footnotesize}
\begin{verbatim}
NDSolve::ndsz: The function value
  {(-1.31378 + 0.I) + 67.5551
     hillm[0.360847 + 0.I, 11.29, 3.50918]
     hillm[52.2162 + 0.I, 17.09, 3.50918], 49, 30}
is not a list of numbers with dimensions {80} at
{t, x2[t], x1[t], x3[t], ...} =
{52.6436, 2.66594+0.I, 108.365+0.I, 0.360847+0.I,
 5.7462e-15 - 2.0572e-69 I, ...}
\end{verbatim}
\end{footnotesize}
\smallskip

\noindent
Three observations from this output are critical.

\emph{First}, the warning fires at $t \approx 52.64$, the moment at which
\texttt{NDSolve}'s function-value check first returns a non-real result: the
production term for $x_2$ has become complex because $x_3$ is being carried as
a complex object ($x_3 = 0.360847 + 0 \cdot i$, i.e.\ real part $0.360847$,
zero imaginary part, but typed as a complex number).

\emph{Second}, the state vector at $t \approx 52.64$ already shows
$x_4 = 5.7 \times 10^{-15} - 2.1 \times 10^{-69}\,i$. The imaginary part
is vanishingly small ($\sim 10^{-69}$) but strictly nonzero, proving that
$x_4$ crossed zero and entered complex arithmetic silently at some earlier time.
The \texttt{ndsz} warning is therefore not the \emph{onset} of contamination;
it is merely the first moment at which imaginary parts grow large enough to
render the function value detectably non-real.

\emph{Third}, this exact warning---same state, same values---is emitted in
\emph{every single \texttt{NDSolve} call} in the notebook: in \texttt{In[1293]}
(the main integration), in \texttt{In[1294]} and \texttt{In[1296]} (the two
plotting calls), and in each of the five \texttt{ExtractExperience} calls
\texttt{In[1298]--[1302]}. Because \texttt{ExtractExperience} calls
\texttt{NDSolve} afresh for each extraction rather than reusing a stored
solution, the solver traverses the same corrupted path every time, and the
warning is reproducibly triggered at $t \approx 52.64$ in all seven independent
runs.

\medskip
\noindent\textbf{A smooth plot does not imply a correct solution.}
The most misleading aspect of the Hill simulation is that the trajectories in
Figure~\ref{fig:hill_zoom} look visually plausible for $t \in [0, 52]$: the
curves are smooth, bounded, and qualitatively similar to sigmoid transients.
This visual appearance cannot be used as evidence of correctness.

From the moment $x_4$ first overshot zero, the right-hand side
$\mathbf{F}(\mathbf{x}(t))$ became a function of complex-valued inputs. The
\texttt{hillm} and \texttt{hillp} functions returned complex outputs whose real
parts were silently treated as real concentrations and carried forward into the
next integration step. The solver's adaptive step-size controller saw a smooth,
well-behaved function---because complex arithmetic is smooth---and raised no
alarm. The result is a set of perfectly smooth curves that faithfully solve a
\emph{complex-corrupted surrogate system}, not the biological ODE. The
corruption is invisible to the eye and detectable only by reading the solver's
message log. This distinguishes the Hill failure fundamentally from ordinary
numerical instability such as step-size blow-up or Runge--Kutta divergence,
where the plot shows visible artefacts---oscillations, spikes, or
explosions---that alert the modeller. Here the plot shows smooth, plausible
dynamics right up to the moment of final collapse at $t \approx 63$--$65$.

\subsubsection{Mechanism 2: Loss of Smoothness and Domain Truncation}
For $n \in (k, k+1)$, the $(k{+}1)$-th derivative of $h^+$ diverges at
$x = 0$, so the ODE vector field is only $C^k$. In our experiment $k = 3$ and
the fourth derivative is singular. Adaptive solvers detect this through
inflating higher-order error estimates and respond by reducing the step size.
As demonstrated above, \texttt{NDSolve} enters complex arithmetic at
$t \approx 52.64$ but does \emph{not} immediately halt: the warning is a
diagnostic, not a stopping condition, and the solver continues building the
\texttt{InterpolatingFunction} until complex arithmetic fully overwhelms
step-size control at approximately $t \approx 63$--$65$, as visible in
Figure~\ref{fig:hill_zoom}. Querying the returned object outside its domain
produces silent polynomial extrapolation:

\smallskip
\begin{footnotesize}
\begin{verbatim}
InterpolatingFunction::dmval:
  Input value {100} lies outside the range of data in the
  interpolating function. Extrapolation will be used.
\end{verbatim}
\end{footnotesize}
\smallskip

\noindent
This \texttt{dmval} warning appears at the extraction queries $t^* = 100$ and
$t^* = 150$, but not at $t^* = 10$, $30$, or $50$, indicating that the
solver's domain extends to at least $t = 50$ but terminates before $t = 100$.
Crucially, the absence of a \texttt{dmval} warning does \emph{not} imply that
the returned values are reliable solutions to the true ODE system. The
\texttt{ndsz} complex-arithmetic warning was emitted \emph{during the
integration itself}, not during post-hoc querying: the piecewise polynomial
stored in the \texttt{InterpolatingFunction} is a faithful record of a
\emph{corrupted} numerical trajectory. Querying it at $t^* = 10$, $30$, or $50$
faithfully interpolates the corrupted track; the returned values are
self-consistent arithmetic but are not solutions to the biological model.

Two failure strata therefore coexist in the Hill experiment. The first is \emph{trajectory corruption}, which affects all five extractions $t^* \in \{10, 30, 50, 100, 150\}$: the numerical path is contaminated by complex arithmetic from the first step in which any concentration overshot zero, so that none of the five extracted values corresponds to a solution of the true ODE system. Concrete evidence is visible even at the earliest times: at $t^* = 10$, $x_4 \approx 1.45 \times 10^{-7}$ (having started from $x_4(0) = 4.58$); at $t^* = 30$ and $t^* = 50$, $x_3 = 0.360847$ exactly, matching the argument in the \texttt{ndsz} warning, which is not a coincidence but a reflection of $x_3$ having converged to its corrupted steady state along the contaminated path. Several other variables exhibit values at $t^* = 10$ (for example $x_{19} \approx 259$, $x_{35} \approx 266$, $x_{51} \approx 208$) that are wildly inconsistent with the corresponding logistic trajectories. The second stratum is \emph{additional domain extrapolation}, affecting only the queries at $t^* = 100$ and $t^* = 150$: beyond trajectory corruption, these two queries lie outside the \texttt{InterpolatingFunction} domain (which ends at approximately $t \approx 63$--$65$) and therefore invoke unconstrained polynomial extrapolation. At $t^* = 100$ this already produces large unphysical values such as $x_5(100) \approx -186$, $x_{12}(100) \approx -10{,}461$, $x_{19}(100) \approx -25{,}493$, and by $t^* = 150$ the divergence is extreme: $x_{12}(150) \approx -81{,}713$, $x_{48}(150) \approx -51{,}998$, while other variables explode in the positive direction ($x_{46}(150) \approx 652{,}260$, $x_{28}(150) \approx 135{,}637$).
The solver fails to reach $t = 200$; the \texttt{InterpolatingFunction} domain
covers approximately $[0, 63\text{--}65]$, leaving roughly $67$--$68\%$ of
the intended horizon extrapolated. The trajectories in
Figures~\ref{fig:hill_zoom} and~\ref{fig:hill_full} make this visible: the
zoomed panel (Figure~\ref{fig:hill_zoom}) pinpoints the onset of instability,
while the full-range panel (Figure~\ref{fig:hill_full}) reveals the subsequent
unbounded growth. The true system would remain positive and bounded by
$\kappa_i / \gamma_i$.

Two further properties of the Hill function amplify the impact of the mechanisms above. First, the slope at the half-activation point $x=\theta$ equals $n/(4\theta)$, coupling steepness and threshold so that they cannot be independently adjusted (the true maximal slope, attained at the inflection point $x=\theta\,((n-1)/(n+1))^{1/n}\neq\theta$, is somewhat larger but exhibits the same coupling). Second, $h^+(0;\,\theta,n) = 0$ for all $n > 0$, so a gene at zero expression contributes nothing to its own production, making the zero state absorbing under Hill dynamics and contradicting observed leaky transcription. Both defects are absent from the logistic formulation~\eqref{eq:logistic}.

\subsection{Stability of the Logistic-Based ODE System}
\label{sec:logistic_stability}

Under the logistic substitution~\eqref{eq:logistic}, all three pathologies disappear simultaneously. The logistic function is globally $C^\infty$ on $\mathbb{R}$, with uniformly bounded derivatives $|f^{\pm\,\prime}| \leq \lambda/4$ (which equals $n/(4\theta)$ under the parameter matching $\lambda = n/\theta$ used here), so the ODE right-hand side inherits this regularity and adaptive solvers can take large time steps throughout the integration. Since $f^\pm$ involves only $\exp(\cdot)$, it is well-defined and real-valued for all $x \in \mathbb{R}$, including negative values arising from numerical overshoot, so no complex-arithmetic contamination is possible. Under the matching $\lambda = n/\theta$, the basal value is $f^+(0;\,\theta,n) = 1/(1+e^{\lambda\theta}) = 1/(1+e^n) > 0$, ensuring that every gene has a strictly positive basal production rate even when all activators are absent, in agreement with observed leaky transcription. Finally, the threshold $\theta$ and steepness $\lambda$ are independent: sharpness can be increased without moving the inflection point.

In the logistic experiment (notebook~\texttt{S2}), \texttt{NDSolve} emits
\emph{no} warnings, completes the integration over $[0,200]$ without
interruption, and the extraction queries at $t^* \in \{10, 30, 50, 100, 150\}$
succeed without extrapolation. All 80 state variables remain non-negative
throughout, and the extracted values are consistent across the five snapshots:
for example, $x_{15}$ starts at $73.4$ at $t = 0$, converges to approximately
$67.9$ by $t = 50$ and holds stably; $x_{22}$ grows to approximately $263$ and
stabilises. Two variables exhibit non-trivial long-run dynamics: the blue curve
$x_2$ oscillates near~$150$ and the orange curve displays sustained periodic
oscillations---a dynamical feature correctly captured without numerical
artefacts.

Figures~\ref{fig:hill_full} and~\ref{fig:hill_zoom} display the Hill-function
trajectories over the full horizon $[0,200]$ and over the early window $[0,65]$
respectively, while Figure~\ref{fig:logistic} shows the logistic counterpart
over $[0,200]$. All panels use the same $n = 3.50918$ and $\theta$ parameters.
Red dash-dotted vertical lines mark the prescribed observation times; the
full-range panels show all five lines ($t^* = 10$, $30$, $50$, $100$, $150$),
while the zoomed panel shows the three lines within its window
($t^* = 10$, $30$, $50$).

\begin{figure}[t]
  \centering
  \includegraphics[width=\textwidth]{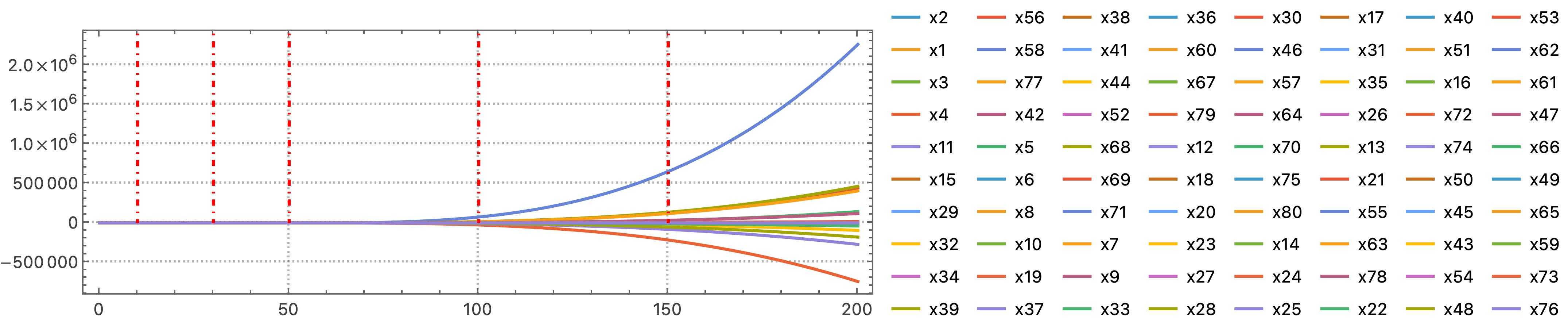}
  \caption{%
    \textbf{Hill-function ODE system ($n = 3.50918$, non-integer),
    full time horizon $t \in [0, 200]$.}
    Trajectories of all 80 state variables. The solver's reliable domain ends
    between $t = 50$ and $t = 100$: beyond that point several curves diverge
    catastrophically. The $y$-axis spans approximately
    $[-0.8\times10^6,\; 2.3\times10^6]$; the most divergent variable
    (blue curve, $x_2$) climbs past $2\times10^6$ by $t = 200$
    (reaching ${\approx}2.3\times10^6$), while the most negative variable
    descends to about $-7.5\times10^5$. Multiple \texttt{NDSolve} warnings about complex-valued
    function evaluations (\texttt{hillm[0.360847 + 0.I, 11.29, 3.50918]}) are
    emitted throughout; \texttt{InterpolatingFunction::dmval} extrapolation
    errors appear at $t^* = 100$ and $t^* = 150$. Red dash-dotted lines mark
    the five observation times $t^* = 10$, $30$, $50$, $100$, $150$.%
  }
  \label{fig:hill_full}
\end{figure}

\begin{figure}[t]
  \centering
  \includegraphics[width=\textwidth]{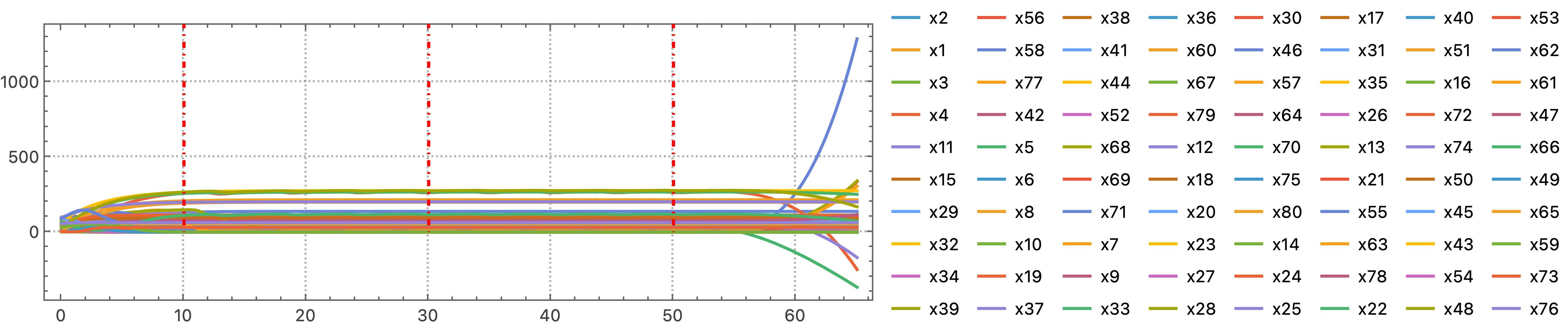}
  \caption{%
    \textbf{Hill-function ODE system ($n = 3.50918$, non-integer),
    early window $t \in [0, 65]$.}
    The same simulation as Figure~\ref{fig:hill_full}, restricted to $[0, 65]$
    to resolve the onset of instability. From $t = 0$ to approximately $t = 52$
    almost all variables remain within biologically plausible ranges. The
    \texttt{NDSolve::ndsz} complex-arithmetic warning fires at $t \approx 52.64$,
    at which point the state already reveals that $x_4$ carries a tiny imaginary
    component ($\sim 10^{-69}$), confirming that complex contamination began
    silently at an earlier time. Integration continues until around
    $t = 63$--$65$, when one variable (blue, $x_2$) begins an exponential ascent
    while at least one other (green curve at the bottom) crosses zero and turns
    negative, marking the final collapse of the solver's domain. Three red
    dash-dotted lines at $t^* = 10$, $30$, and $50$ fall within this window;
    the remaining two ($t^* = 100$, $150$) lie off-screen to the right.
    Crucially, the visual smoothness of the curves for $t \in [0, 52]$ is not
    evidence of correctness: the solver is faithfully integrating a
    complex-corrupted surrogate system. All five extractions are unreliable; the
    two off-window queries additionally invoke domain extrapolation.%
  }
  \label{fig:hill_zoom}
\end{figure}

\begin{figure}[t]
  \centering
  \includegraphics[width=\textwidth]{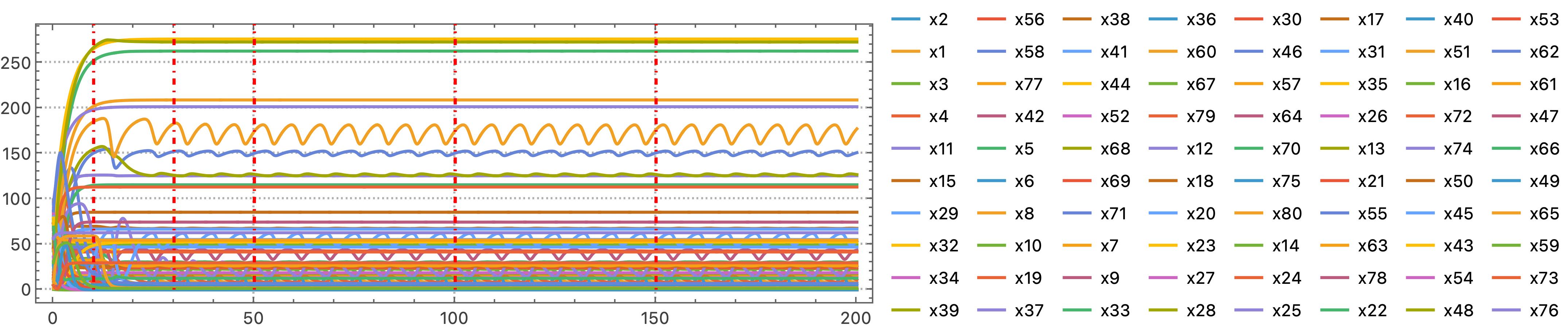}
  \caption{%
    \textbf{Logistic-function ODE system (same $n$ and $\theta$ parameters),
    full time horizon $t \in [0, 200]$.}
    Trajectories of all 80 state variables. \texttt{NDSolve} completes the
    integration without any warnings. All variables remain strictly non-negative
    throughout and converge to bounded steady states; the $y$-axis is confined to
    $[0, 275]$, consistent with the biological bound $\kappa_i/\gamma_i$. Most
    variables settle before $t = 50$. Two variables exhibit non-trivial long-run dynamics: the blue curve $x_2$ oscillates near $150$ and the orange curve displays sustained periodic oscillations, a dynamical feature correctly captured without numerical artefacts. All five extractions at
    $t^* \in \{10, 30, 50, 100, 150\}$ lie within the solver's domain and yield
    physically meaningful values.%
  }
  \label{fig:logistic}
\end{figure}

\subsection{Hill Functions with Real Exponent: A Generically Unreliable Framework}
\label{sec:general_failure}

The experiment reported above is not an isolated unlucky draw. The failures documented in Section~\ref{sec:hill_pathologies_empirical} are \emph{structural} consequences of the mathematical definition of $h^\pm$, and the conditions producing them are present whenever $n \notin \mathbb{N}$. Fitted Hill coefficients in the systems biology literature consistently report non-integer values (see, e.g.,~\citep{alon, ingalls2013mathematical}): transcription-factor binding curves typically yield $n \in [1.2, 3.8]$; cooperative enzyme kinetics give $n \approx 1.7$--$2.4$; and synthetic toggle switches and repressilators are fitted with $n \approx 2.1$--$4.6$. Our simulation draws $n$ from $\mathrm{U}(1,5)$, so the probability of drawing an integer is exactly zero, and the fixed value $n = 3.50918$ is representative of the generic situation encountered in practice.

The various failure modes are manifestations of a single underlying cause: the ODE system derived from Hill functions with real-valued $n$ is \emph{numerically unstable} in a precise sense. A numerical integration method applied to $\dot{\mathbf{x}} = \mathbf{F}(\mathbf{x})$ is locally stable only if small perturbations do not grow unboundedly, which requires, at minimum, that $\mathbf{F}$ be Lipschitz-continuous in a neighbourhood of the trajectory~\citep{hairer1993solving}. For $n \in (k, k+1) \subset \mathbb{R}$ ($k \in \mathbb{N}_0$), the right-hand side $\mathbf{F}$ fails to be sufficiently smooth on a neighbourhood of the boundary of the positive orthant in two distinct, compounding ways. For $0 < n < 1$ (the case $k=0$), the partial derivative of $F_i$ with respect to any of its regulator variables $x_j$ diverges as $x_j \to 0^+$, so $\mathbf{F}$ is not locally Lipschitz at the boundary; the Lipschitz constant $L = \sup\|\nabla\mathbf{F}\|$ is infinite near the origin, standard error bounds of the form $\|\mathbf{e}(t)\| \leq \|\mathbf{e}(0)\|\,e^{Lt}$ give no useful guarantee, and the classical existence--uniqueness theorem itself does not apply directly. For $n > 1$ (the case $k \geq 1$), $\mathbf{F}$ is locally Lipschitz but only $C^{\lfloor n \rfloor}$: the $(k{+}1)$-th derivative of $F_i$ with respect to any regulator $x_j$ diverges as $x_j \to 0^+$, so standard convergence theorems for Runge--Kutta methods of order $p > k$ cease to apply near the boundary, and adaptive solvers that estimate higher derivatives produce spuriously large local error estimates and are forced to reduce step size dramatically. In both regimes, moreover, once any $x_i$ overshoots to a negative value---even by floating-point rounding of order $10^{-15}$---the expression $x_i^n$ is no longer a real number, and the trajectory leaves the domain on which $\mathbf{F}$ is defined as a real-valued vector field. That $x_i^n$ is non-real for $x_i<0$ and non-integer $n$ is universal: it holds in every computing environment. The resulting \emph{failure mode}, however, depends on how the environment evaluates fractional powers of negative numbers. Environments that return a value on the complex principal branch---\textit{Mathematica}, and Python's built-in \texttt{**} operator---propagate the contamination silently, exactly as documented above. Environments that instead return \texttt{NaN}, notably NumPy's \texttt{power}, cause the affected step to be rejected by the adaptive error controller, so the failure surfaces noisily---as severe step-size reduction or stalling---rather than as silent corruption. The silent, artefact-free corruption analysed in this paper is therefore characteristic of complex-propagating environments; the underlying obstruction---that the real-valued Hill field is undefined past a zero crossing---is not.

The positive-orthant boundary is precisely the region that low-expression states and transient dynamics explore. As a consequence, standard convergence and stability theorems for Runge--Kutta and multistep methods do not apply, and the computed solution carries no guaranteed accuracy bound even in exact arithmetic. This loss of the convergence guarantees is independent of the programming language, tolerance, or algorithm: it follows from the limited smoothness of the Hill field itself. In finite-precision arithmetic, a trajectory that passes through a neighbourhood of zero---virtually inevitable for a randomly initialised $80$-dimensional system---additionally triggers the fractional-power obstruction described above, with the environment-dependent consequences identified there. By contrast, each logistic factor in~\eqref{eq:logistic} is globally $C^\infty$ and globally Lipschitz (with constant $\lambda/4$), so the right-hand side of the multi-gene ODE system is globally $C^\infty$ and globally Lipschitz on $\mathbb{R}^N$ with the explicit constant recalled in Section~\ref{sec:prelim} (established in the companion paper~\citep{belgacem2026framework}), and standard stability theory applies everywhere, including near and below zero.

The pathology compounds with network size and observation horizon. In a small two- or three-variable system, solver failure is easy to detect visually; but in a high-dimensional network ($N = 80$ here, with realistic GRN models routinely reaching $N > 100$), the failure of a single variable to remain non-negative suffices to corrupt the entire right-hand side simultaneously, and the larger the network, the more likely that at least one trajectory passes through or near zero during the transient phase. Longer observation horizons increase this probability further. In our experiment the \texttt{InterpolatingFunction} domain ends at $t \approx 63$--$65$ out of $t_{\max} = 200$, leaving roughly $67$--$68\%$ of the intended horizon extrapolated; but this is the secondary failure, the primary pathology being the trajectory corruption that begins silently, well before the visible \texttt{ndsz} warning, and that makes all five extracted observations unreliable.

Practitioners sometimes attempt to circumvent these problems through ad~hoc patches. Some address an individual symptom, but none restores the full set of properties that make a kernel reliable. \emph{Evaluating the Hill kernel at the clamped argument $\max(x_j,0)$} does remove the complex-arithmetic contamination---the kernel is then never asked for a fractional power of a negative number---and, unlike $|x|^n$ below, the resulting function remains monotone in $x_j$. It does not, however, remove the absorbing off-state, since $h^+(\max(0,0);\theta,n)=h^+(0;\theta,n)=0$; nor does it improve smoothness, since the clamped field is still only $C^{\lfloor n\rfloor}$ at the orthant boundary, so the order loss of Proposition~\ref{prop:apriori}(ii) persists; and it is a non-biological implementation layer that the logistic kernel does not require. \emph{Replacing $x^n$ by $|x|^n$} likewise restores real-valuedness but makes $h^+$ an even function near zero, destroying the monotone sigmoid shape that gives the Hill function its biological meaning. \emph{Clamping the integrator state to $[0,\infty)^N$ after each step} projects away negative excursions but introduces a non-smooth correction that interferes with adaptive error control and itself blocks gradient-based parameter estimation. \emph{Adding a small offset $\epsilon > 0$ to $x$} removes the singularity at the origin but introduces an arbitrary parameter with no biological interpretation, shifts the effective threshold, and breaks the normalisation $h^+ \in [0,1]$. \emph{Rounding $n$ to the nearest integer} changes the fitted parameter value, potentially moves the system across a bifurcation boundary, and is epistemically unjustified when $n$ is determined by experimental data. Finally, \emph{tightening solver tolerances} (reducing \texttt{AccuracyGoal} and \texttt{PrecisionGoal}) postpones but does not prevent the failure: a complex-propagating environment will still evaluate $x^n$ once $x$ becomes negative, and the $C^{\lfloor n\rfloor}$ order loss is independent of tolerance. The recurring point is that each patch trades one defect for another, whereas the logistic kernel has none of these defects to begin with.

The theoretical analysis and numerical experiments of this section therefore converge on a single conclusion: \emph{Hill functions with non-integer cooperativity exponent are generically unsuitable as regulatory kernels in ODE models derived from Boolean networks.} The logistic substitution, with the parameter correspondence $\lambda = n/\theta$, resolves every one of the defects simultaneously: the logistic right-hand side is globally $C^\infty$, real-valued for all arguments, strictly positive at the origin, and has fully decoupled steepness and threshold. Under identical network, parameters, and solver settings, it produces a complete, warning-free integration and physically consistent state extractions at all five observation times. We therefore propose logistic functions as the default regulatory kernel for Boolean-to-ODE translation and, more broadly, for GRN models calibrated to experimental dose-response data, where the cooperativity coefficient is virtually never an integer~\citep{alon, ingalls2013mathematical}.

\subsection{An \emph{A Priori} Error Bound for Logistic-Based Integration}
\label{sec:apriori}

The empirical contrast of
Sections~\ref{sec:hill_pathologies_empirical}--\ref{sec:general_failure} can be
promoted from an observation into a quantitative \emph{a priori} guarantee. For
a numerical-simulation study this is the decisive point: not merely that the
logistic integration completed without warnings, but that its error is
controlled, in advance, by an explicit and computable bound of classical
type---and that for the Hill system no bound of that type can exist.
Proposition~\ref{prop:apriori} makes this precise; it is the rigorous
counterpart of the mechanism analysis of Section~\ref{sec:general_failure}.

\begin{proposition}[Convergence of one-step methods for the logistic GRN system]
\label{prop:apriori}
Let $\dot{\mathbf{x}}=\mathbf{F}(\mathbf{x})$ denote the logistic
product-of-logistics system~\eqref{eq:multi_gene_system} on $\mathbb{R}^{N}$,
whose right-hand side is globally $C^{\infty}$ and globally Lipschitz on
$\mathbb{R}^{N}$ with the explicit constant
\[
  L \;=\; \max_{1\le i\le N}
          \Bigl(\gamma_i+\tfrac14\,\kappa_i\lambda\,
                (|\mathcal{A}_i|+|\mathcal{R}_i|)\Bigr)
\]
recalled in Section~\ref{sec:prelim}.
\begin{enumerate}
\item[(i)] (\emph{Logistic system: full-order convergence with explicit
constants.})  Fix any initial state $\mathbf{x}_0\in\mathbb{R}^{N}_{\ge0}$ and
set $b_i=\max\{x_{0,i},\,\kappa_i/\gamma_i\}$.  The box
$\mathcal{B}'=\prod_{i}[0,b_i]$ is compact and forward invariant, so the exact
solution satisfies $\mathbf{x}(t)\in\mathcal{B}'$ for all $t\ge0$.  Let a
one-step method of order $p$ (for instance a Runge--Kutta method of order $p$)
be applied with constant step size $h$.  Then the global error at $t_n=nh$
satisfies
\begin{equation}
  \bigl\|\mathbf{x}_n-\mathbf{x}(t_n)\bigr\|
  \;\le\;
  \frac{C_p}{L}\;h^{p}\,\bigl(e^{L t_n}-1\bigr),
  \label{eq:apriori_bound}
\end{equation}
where $C_p<\infty$ depends only on the method and on
$\sup_{t\ge0}\|\mathbf{x}^{(p+1)}(t)\|$, the latter being finite because the
trajectory remains in the compact set $\mathcal{B}'$ on which $\mathbf{F}$ is
$C^{\infty}$.  The method therefore attains its full order $p$ uniformly on
$[0,\infty)$, and both $L$ and $C_p$ are computable from
$(\boldsymbol{\kappa},\boldsymbol{\gamma},\lambda)$ \emph{before} the
integration is run, so a step size meeting any prescribed error tolerance can
be chosen a priori.
\item[(ii)] (\emph{Hill system: no bound of the
form~\eqref{eq:apriori_bound}.})  Let $\mathbf{F}_{\mathrm{H}}$ be the
right-hand side produced by the same Boolean-to-ODE translation with the Hill
kernels~\eqref{eq:hill} of non-integer exponent $n$.  If $0<n<1$, then
$\partial F_{\mathrm{H},i}/\partial x_j\to\infty$ as $x_j\to0^{+}$, so
$\mathbf{F}_{\mathrm{H}}$ admits no finite Lipschitz constant on any
neighbourhood of $\partial\mathbb{R}^{N}_{\ge0}$ and the constant $L$
in~\eqref{eq:apriori_bound} is infinite.  If $n>1$, then
$\mathbf{F}_{\mathrm{H}}$ is locally Lipschitz on the open orthant but only of
class $C^{\lfloor n\rfloor}$, so for every method of order
$p>\lfloor n\rfloor$ the local truncation error fails to be $O(h^{p+1})$ near
$\partial\mathbb{R}^{N}_{\ge0}$ and the order-$p$
estimate~\eqref{eq:apriori_bound} does not hold.  In both cases
$\mathbf{F}_{\mathrm{H}}$ is undefined as a real-valued field at any state with
a negative coordinate, so once a numerical overshoot crosses zero the
trajectory leaves the domain on which the convergence theory is even
formulated.
\end{enumerate}
\end{proposition}

\begin{proof}
\emph{(i)}  Each logistic factor in~\eqref{eq:multi_gene_system} lies in
$(0,1)$, so $0<\kappa_i\prod(\cdots)<\kappa_i$ and hence
$-\gamma_i x_i\le\dot{x}_i<\kappa_i-\gamma_i x_i$.  Thus $\dot{x}_i<0$ whenever
$x_i\ge\kappa_i/\gamma_i$ and $\dot{x}_i\ge0$ whenever $x_i\le0$, so every box
$\prod_i[0,b_i]$ with $b_i\ge\kappa_i/\gamma_i$ is forward invariant by
Nagumo's theorem~\citep{blanchini2008set}; the choice
$b_i=\max\{x_{0,i},\kappa_i/\gamma_i\}$ yields a compact forward-invariant set
containing $\mathbf{x}_0$, whence $\mathbf{x}(t)\in\mathcal{B}'$ for all
$t\ge0$.  On the compact set $\mathcal{B}'$ the $C^{\infty}$ field
$\mathbf{F}$ has all derivatives bounded, so
$\mathbf{x}^{(p+1)}=\tfrac{d^{p}}{dt^{p}}\mathbf{F}(\mathbf{x})$ is uniformly
bounded and the order constant $C_p$ is finite.  The two hypotheses of the
classical convergence theorem for one-step methods---a globally Lipschitz field
and a solution with bounded $(p{+}1)$-th derivative---are therefore both met,
and the theorem~\citep{hairer1993solving}, applied with Lipschitz constant $L$
and zero initial error, yields~\eqref{eq:apriori_bound}.

\emph{(ii)}  For the Hill kernel
$h^{+}(x;\theta,n)=x^{n}/(\theta^{n}+x^{n})$,
\[
  \frac{\partial h^{+}}{\partial x}
  =\frac{n\,\theta^{n}x^{n-1}}{(\theta^{n}+x^{n})^{2}}
  \;\sim\;\frac{n}{\theta^{n}}\,x^{\,n-1}
  \qquad(x\to0^{+}).
\]
For $0<n<1$ the exponent $n-1<0$ makes this diverge, so no finite Lipschitz
constant exists near the boundary.  For $n>1$ the first derivative is bounded,
but the $(\lfloor n\rfloor+1)$-th derivative of $x^{n}$ scales as
$x^{\,n-\lfloor n\rfloor-1}\to\infty$, so $\mathbf{F}_{\mathrm{H}}$ is exactly
of class $C^{\lfloor n\rfloor}$ and the solution is then only
$C^{\lfloor n\rfloor+1}$; the Taylor expansion of order $p+1$ underlying an
$O(h^{p+1})$ local truncation error therefore does not exist for
$p>\lfloor n\rfloor$ near $\partial\mathbb{R}^{N}_{\ge0}$.  Finally, for $x<0$
and $n\notin\mathbb{N}$ the principal branch gives
$x^{n}=|x|^{n}e^{i\pi n}\notin\mathbb{R}$, so $\mathbf{F}_{\mathrm{H}}$ is not a
real-valued vector field at such states.
\qedhere
\end{proof}

Proposition~\ref{prop:apriori} is, from a numerical-simulation standpoint, the
theoretical core of the present study.  It establishes that the reliability gap
measured in Section~\ref{sec:numerical_comparison} is not a parameter-specific
accident but a structural fact: the logistic formulation places the
Boolean-derived ODE system squarely inside the hypotheses of the classical
convergence theory of one-step methods, with a Lipschitz constant and an error
constant that are \emph{explicit functions of the network parameters} and can
be evaluated before any integration is performed.  This is the precise,
quantitative content of the word ``robust'': a step size guaranteeing a
prescribed accuracy is available a priori.  The Hill formulation with a fitted,
and therefore generically non-integer, exponent violates these hypotheses
structurally, by part~(ii).  The $80$-gene experiment of
Section~\ref{sec:numerical_comparison} is the empirical shadow of the
proposition---the logistic run realises the guaranteed bounded-error regime
of~\eqref{eq:apriori_bound}, while the Hill run exhibits exactly the breakdown
that part~(ii) predicts---and the proposition shows that the structural
breakdown of part~(ii), namely the unavailability of an \emph{a priori,
uniform} error bound of the form~\eqref{eq:apriori_bound}, must recur for
every Boolean-derived network with non-integer exponent,
independently of solver, language, or tolerance.

\subsection{Numerical Verification of the Convergence Bound}
\label{subsec:convergence_check}

Proposition~\ref{prop:apriori}(i) is an a priori statement: it guarantees,
before any integration is performed, that a one-step method of order $p$
applied to the logistic GRN system incurs a global error $O(h^{p})$. We close
the section by confirming this order numerically and, in doing so, delimiting
precisely what part~(ii) does and does not assert.

We integrate the six-gene logistic motif of
Section~\ref{subsec:minimal_boolean} on $[0,T]$ with $T=6$, using three
explicit one-step methods of known order---the forward Euler method ($p=1$),
the explicit midpoint method ($p=2$), and the classical fourth-order
Runge--Kutta method ($p=4$)---at the constant step sizes $h=T/2^{k}$ for
$k=6,\dots,14$. The global error $\|\mathbf{x}_N-\mathbf{x}(T)\|_2$ at the
endpoint is measured against a high-accuracy reference solution.
Figure~\ref{fig:convergence} plots this error against $h$ on logarithmic axes.

The three error curves are straight, with fitted slopes $1.00$, $2.01$, and
$4.03$, each agreeing with the order $p$ of its method. The Runge--Kutta error
reaches the double-precision rounding floor near $h=3\times10^{-3}$ and then
levels off, as expected of a fixed-step computation. The experiment confirms
the order-$p$ scaling of the bound~\eqref{eq:apriori_bound}: on the logistic
GRN system every method attains its full classical order, so the step size
needed for a prescribed accuracy is the one the proposition specifies in
advance. Repeating the computation in a second language and arithmetic
environment reproduces the same three slopes, as the structural argument
requires.

Running the identical experiment with the Hill kernels of the same exponent
$n=3.50918$ yields, on this trajectory, the same three orders. This does not
contradict part~(ii). The trajectory considered here remains in the open
orthant $\mathbb{R}^{N}_{>0}$, on which the Hill field is real-analytic, so a
method applied to it sees a smooth problem and converges accordingly. What
part~(ii) denies the Hill system is the bound~\eqref{eq:apriori_bound} as an
\emph{a priori and uniform} guarantee: its error constant degrades without
bound as a trajectory approaches $\partial\mathbb{R}^{N}_{\ge0}$, no finite
Lipschitz constant exists there when $n<1$, and the field ceases to be
real-valued the instant a numerical overshoot crosses zero---the qualitative
breakdown documented in Section~\ref{sec:numerical_comparison}. The contrast
established by Proposition~\ref{prop:apriori} is therefore not that the Hill
system fails to converge on benign interior trajectories, but that only the
logistic system carries a convergence guarantee that can be certified before
the integration is run and that holds uniformly up to the boundary of the
non-negative orthant, where Boolean-derived dynamics naturally operate.

\begin{figure}[t]
  \centering
  \includegraphics[width=0.78\textwidth]{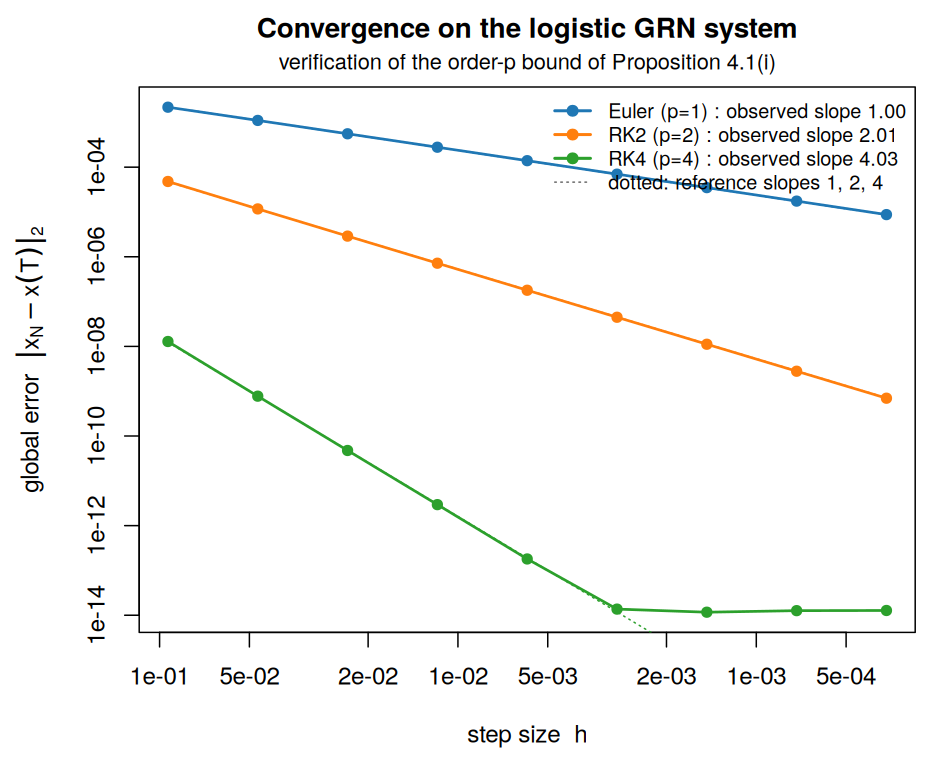}
  \caption{%
    Numerical verification of Proposition~\ref{prop:apriori}(i). Global error
    at $t=T=6$ for the six-gene logistic motif of
    Section~\ref{subsec:minimal_boolean}, integrated by the forward Euler,
    explicit midpoint, and classical Runge--Kutta methods at constant step
    sizes $h=T/2^{k}$, $k=6,\dots,14$, against a high-accuracy reference
    solution. The fitted slopes ($1.00$, $2.01$, $4.03$) match the method
    orders $p=1,2,4$; dotted lines are exact reference slopes. The
    Runge--Kutta curve levels off near $h=3\times10^{-3}$ at the
    double-precision rounding floor.%
  }
  \label{fig:convergence}
\end{figure}

\section{Conclusion}
\label{sec:conclusion}

This paper has presented a simulation study comparing Hill-based and
logistic-based ordinary differential equation models of gene regulatory
networks, with a focus on the numerical reliability of integrating
Boolean-derived ODE systems and on the prevention of irreversible expression
shutdown. The evidence assembled here, across an experimentally grounded
single-gene motif and an $80$-gene benchmark, supports a single conclusion: the
Hill function is a generically unreliable kernel for large-scale GRN
simulation, and the logistic function is a robust replacement that requires no
specialised numerical machinery.

The low-expression analysis of Section~\ref{sec:biological} isolated the first
failure mode. Because the Hill activation function vanishes identically at zero
input, a bistable positive-autoregulation circuit acquires an absorbing
off-state. Theorem~\ref{thm:bistability} characterised the saddle-node set of
the logistic model through the explicit transcendental equation $e^{z}-z =
\lambda\theta-1$, identified the threshold $\lambda\theta=2$ separating
monostable from bistable regimes, and showed that the bistable window widens
asymptotically as $(\theta,\,e^{\lambda\theta-1}/\lambda)$;
Corollary~\ref{cor:autoreg_stability} classified the stability of every
equilibrium and certified global asymptotic stability in the monostable cases.
Using biophysically grounded \textit{E.~coli} parameters
(Table~\ref{tab:parameters}), numerical simulation showed that the logistic
model escapes the off-state in approximately $44$~minutes through basal
production alone---consistent with a conservative linear analytical estimate of
${\approx}58$~minutes and with galactose-operon induction kinetics---whereas the
Hill model produces mRNA synthesis more than three orders of magnitude below the
degradation rate at the initial condition, so its trajectory reaches a transient
peak well below the unstable separatrix at $x^*\approx0.041$ and remains
permanently confined near zero with no intrinsic recovery mechanism.

The large-scale experiment of Section~\ref{sec:numerical_comparison} provided
the most direct evidence of the reliability gap. On an $80$-gene
Boolean-derived ODE system with non-integer Hill exponent $n\approx3.509$, the
Hill-based solver entered silent complex-valued arithmetic contamination from
the first moment any state variable overshot zero---with imaginary components of
order $10^{-69}$ appearing well before the visible \texttt{NDSolve::ndsz}
warning fired at $t\approx52.64$---and produced smooth, visually plausible
trajectories that were solutions of a corrupted surrogate system rather than the
true biological model. Integration terminated near $t\approx63$--$65$, leaving
roughly $67$--$68\%$ of the intended horizon covered only by unconstrained
polynomial extrapolation, with biologically impossible concentrations among the
extrapolated values, and no visual artefact to signal the failure. Under
identical parameters and initial conditions, the logistic formulation completed
the full integration over $t\in[0,200]$ without a single warning, with all $80$
state variables strictly non-negative and bounded throughout. The decisive
structural difference is that each logistic factor is globally $C^\infty$ and
globally Lipschitz, so the right-hand side of the logistic ODE system satisfies
the hypotheses of the standard convergence and stability theory of Runge--Kutta
and multistep methods everywhere in state space, including near and below zero,
whereas the Hill right-hand side does not.
Proposition~\ref{prop:apriori} sharpens this structural statement into a
quantitative one: it provides an explicit a priori global-error bound for
logistic-based integration, with Lipschitz and error constants computable
from the network parameters before integration begins, and proves that no
bound of this form can exist for the Hill system.

The framework studied here is ready for immediate deployment. Implementations
require only standard numerical integration libraries, and the logistic
structure is natively compatible with the automatic-differentiation tools used
in modern machine-learning platforms such as TensorFlow, PyTorch, and JAX.
Natural directions for future work include extension to stochastic formulations
that account for intrinsic noise in low-copy-number regimes, the incorporation
of spatial dynamics and cell-to-cell communication in multicellular systems,
and the development of reduced-order surrogates for accelerated large-scale
simulation. By replacing Hill functions with their logistic counterparts while
preserving sigmoidal dynamics, modellers retain decades of accumulated
Hill-based modelling intuition while gaining the numerical reliability that
demanding large-scale simulation requires.

\section*{Statements and Declarations}

\subsection*{Competing interests}
The author declares that he has no competing financial interests or
personal relationships that could have appeared to influence the work
reported in this paper.

\subsection*{Funding}
This research received no specific grant from any funding agency in the
public, commercial, or not-for-profit sectors.

\subsection*{Data availability}
The low-expression genetic-oscillator comparison of
Section~\ref{sec:biological} (Figure~\ref{fig:Oscillateur_original_l_H}) and
the positive-autoregulation analysis (Section~\ref{ex:positive_autoregulation},
Figure~\ref{fig:autoreg_sim}) were conducted in R using the \texttt{deSolve}
package's \texttt{ode} function, as was the Boolean-derived six-gene motif of
Section~\ref{subsec:minimal_boolean} (Figure~\ref{fig:boolean_motif}). The positive-autoregulation analysis employed
the biophysically grounded \textit{E.~coli} parameters documented in
Table~\ref{tab:parameters}: $k_m = 0.003$~molecules\,s$^{-1}$, $k_{dm} = 0.001$~s$^{-1}$,
$k_p = 0.002$~s$^{-1}$, $k_{dp} = 0.00001$~s$^{-1}$, $\lambda = n = 3$,
$\theta = c = 1$, with initial conditions $m(0) = x(0) = 0.01$. The $80$-gene
Boolean-derived ODE experiments of Section~\ref{sec:numerical_comparison} were
implemented in \textit{Mathematica}; Mathematica notebooks implementing
the $80$-gene Hill and logistic experiments (\texttt{S1}, \texttt{S2}) are provided as supplementary
material, with all parameters hard-coded for full reproducibility. The convergence
verification of Section~\ref{subsec:convergence_check}
(Figure~\ref{fig:convergence}) was carried out independently in R and in
Python, the two implementations agreeing to the fitting precision. No
experimental datasets were generated; all parameter values are drawn from
published literature.

\bibliographystyle{elsarticle-num}
\bibliography{mybibfile_corrected}

\end{document}